\documentclass[english]{article}
\usepackage[T1]{fontenc}
\usepackage[latin9]{inputenc}
\usepackage{geometry}
\usepackage{url}
\usepackage{amsmath}
\usepackage{verbatim}
\usepackage{amsfonts}
\usepackage{xr}
\usepackage{enumerate}
\usepackage[T1]{fontenc}
\usepackage[latin9]{inputenc}
\usepackage{color}
\usepackage{array}
\usepackage{float}
\usepackage{multirow}
\usepackage{graphicx}
\usepackage{enumitem}
\usepackage{booktabs}
\usepackage{centernot}
\usepackage{hyperref}
\usepackage{natbib}
\usepackage{amsthm}
\usepackage{algorithm}

\newtheorem*{definition*}{Definition}

\newtheorem{assumption}{Assumption}

\newtheorem{theorem}{Theorem}

\newtheorem{lemma}{Lemma}

\newtheorem{remark}{Remark}

\newcommand*{\de}{\mathrm{d}}
\newcommand*{\N}{\mathcal{N}}
\newcommand{\R}{\mathbb{R}}
\newcommand{\T}{\mathrm{\scriptscriptstyle T}}
\newcommand*{\E}{\mathrm{E}}
\newcommand{\var}{\mathrm{var}}
\newcommand{\cov}{\mathrm{cov}}

\newcommand{\pr}{\mathrm{pr}}
\newcommand{\nai}{\mathrm{nai}}
\newcommand{\ipw}{\mathrm{ipw}}
\newcommand{\aipw}{\mathrm{aipw}}

\newcommand{\psm}{\mathrm{psm}}
\newcommand{\xm}{\mathrm{x.m}}
\newcommand*{\sumin}{\sum_{i=1}^n}
\newcommand*{\intin}{\int_{0}^\tau}
\def\bSig\mathbf{\Sigma}

\newcommand{\indep}{\rotatebox[origin=c]{90}{$\models$}}

\usepackage{authblk}

\title{Propensity score matching for estimating a marginal hazard ratio}
\author[1]{Tongrong Wang\footnote{Contributed equally.}}
\newcommand\CoAuthorMark{\footnotemark[\arabic{footnote}]}
\author[2]{Honghe Zhao\protect\CoAuthorMark}
\author[2]{Shu Yang\thanks{syang24@ncsu.edu}}
\author[3]{Shuhan Tang}
\author[1]{Zhanglin Cui}
\author[4]{Li Li}
\author[1]{Douglas E. Faries}
\affil[1]{Eli Lilly and Company, Indianapolis, Indiana, USA}
\affil[2]{Department of Statistics, North Carolina State University, Raleigh, North Carolina, USA}
\affil[3]{Department of Statistics, The Ohio State University, Columbus, Ohio, USA}
\affil[4]{Statistical Science, R\&G US Inc, Somerset, New Jersey, USA}

\begin{document}
\maketitle





\begin{abstract}
Propensity score matching is commonly used to draw causal inference from observational survival data. However, its asymptotic properties have yet to be established, and variance estimation is still open to debate. We derive the statistical properties of the propensity score matching estimator of the marginal causal hazard ratio based on matching with replacement and a fixed number of matches. We also propose a double-resampling technique for variance estimation that takes into account the uncertainty due to propensity score estimation prior to matching.
\end{abstract}


\section{Introduction}
Survival analysis plays an increasingly important role in treatment effect estimation due to the frequent occurrence of time-to-event outcomes in biomedical studies. By comparing the hazard functions of survival times between treated and untreated individuals, the marginal hazard ratio is commonly used to measure the effect of treatment on a time-to-event outcome for a particular population of interest. The log of marginal hazard ratio corresponds to the coefficient indexing a univariate Cox proportional hazard model, where the hazard of the outcome is regressed on an indicator denoting treatment status \citep{cox1972regression}.

In observational studies, propensity score (PS) methods are standard approaches for reducing the effect of confounding. However, depending on the specific type or implementation of the propensity score methods, the population parameters or treatment effects being estimated may not be the same. In general, a measure of treatment effect can be classified as conditional or marginal. Conditional effects correspond to an average effect at the individual level, whereas marginal effects denote an average effect at the population level. Hazard ratios are non-collapsible measures, which means that the marginal and conditional hazard ratios generally do not coincide even in the absence of confounding \citep{martinussen2013collapsibility}. Therefore, the non-collapsibility of the hazard ratios renders methods that estimate the conditional treatment effects biased for estimating the marginal hazard ratios \citep{austin2007conditioning,austin2013performance}. For instance, stratification on the propensity score and covariate adjustment using the propensity score estimate conditional treatment effects, so they generally yield biased estimates of the marginal hazard ratio \citep{williamson2012propensity, vansteelandt2014regression}. The inverse probability weighted (IPW) estimator fits a marginal Cox model with each observation weighted by the reciprocal of its estimated probability of receiving the observed treatment. This estimator is consistent for the marginal hazard ratio if the propensity score model is correctly specified and the overlap assumption is satisfied. However, it is sensitive to slight misspecification of the PS model and can yield large variance when the estimated PS is close to 0 \citep{kang2007demystifying,busso2014new,austin2017performance}. To protect against misspecification of the propensity score model, the doubly robust augmented inverse probability weighted (AIPW) estimator is proposed, which combines IPW with an outcome regression model \citep{tchetgen2012parametrization}. It is consistent as long as either the outcome model or the propensity score model is correctly specified. However, non-collapsibility of the hazard ratios makes specifying a correct outcome model increasingly difficult since it needs to marginalize to a survival curve that satisfies the marginal proportional hazard assumption. Literature also suggests the performance of AIPW in finite sample scenarios can be unstable when both models of the PS and the outcome are misspecified and are also sensitive to extreme values of the estimated PS \citep{kang2007demystifying,waernbaum2012model}.

Alternatively, matching methods enjoy multiple desirable features. First, matching is transparent and has great intuitive appeal as it seeks to emulate a completely randomized experiment using observational data \citep{dehejia2002propensity,rubin2006matched,stuart2010matching}. Second, empirical evidence suggests that while matching on the PS and IPW do not uniformly outperform one another in all situations, PS matching tends to be more robust to misspecification of the PS model and to extreme values of the estimated PS (practical violation of the overlap assumption) \citep{frolich2004finite,waernbaum2012model,busso2014new,austin2020variance,greifer2021matching}. Moreover, matching does not rely on an outcome model and thus avoids the aforementioned congeniality issue between the outcome model and the proportional hazard assumption when applying AIPW. Simulation results have shown that greedy nearest neighbor matching on the propensity score without replacement results in unbiased estimation of the marginal hazard ratio over the subpopulation of treated individuals \citep{austin2013performance}. However, when the amount of treated and control units are comparable, bias could arise due to incomplete matches. PS matching with replacement not only circumvents this issue but also permits estimation of the marginal treatment effect on the overall population containing both treated and untreated subjects \citep{austin2020variance,abadie2006large,williamson2012propensity}.

While PS matching with replacement is potentially attractive for estimating the marginal hazard ratio, no formal asymptotic results have been established. For variance estimation, most existing approaches do not take into account the uncertainty of parametrically estimating the propensity score prior to matching and restrict inference to the matched sample \citep{gayat2012propensity,austin2013performance,austin2020variance}. Moreover, although \citet{austin2014use} demonstrated that bootstrap is valid for matching without replacement, it is inappropriate for matching with replacement because the bootstrap sample fails to preserve the distribution for the number of times each individual is used as a match \citep{abadie2008failure}. Thus, it is important to develop a variance estimation procedure for the matching with replacement estimator of the marginal hazard ratio.

This article concerns propensity score matching with replacement and with a fixed number of matches for estimating the marginal hazard ratio for a point binary treatment given pre-treatment covariates. We will simply refer to this procedure as PSM or the PSM estimator. We note that PSM involves imputing the missing potential outcome processes, and is distinct from paired (one-to-one) matching without replacement, which could result in dropping units from the analysis \citep{abadie2006large}.

In this article, we first derive the large sample distribution of the PSM estimator based on known propensity scores. Secondly, because the propensity score function is often estimated prior to matching, we also derive the large sample distribution of PSM accounting for the uncertainty due to nuisance parameter estimation. Since PSM is a nonsmooth functional of the distribution of the data, our derivation is based on the technique developed by \citet{andreou2012alternative}, which offers a general recipe for deriving the limiting distribution of nonsmooth statistics that involve estimated nuisance parameters. This technique has been successfully applied by \citet{abadie2016matching} for estimating the average treatment effects for a continuous outcome. Our derivation extends their results to the survival context. This extension is not trivial because the survival outcome is often right-censored. We utilize the martingale theory of the counting process to establish asymptotic distributions of the PSM estimator of the marginal hazard ratio. In addition, we propose a replication-based variance estimator for PSM that accounts for the uncertainty of nuisance parameter estimation. For a continuous outcome, \citet{adusumilli2022bootstrap} proposed a bootstrap inference procedure, which improves upon the asymptotic variance estimator proposed by \citet{abadie2016matching}. Our proposed method extends Adusumilli's technique to the survival context. Simulation results suggest that it generally outperforms Wald-type inference based on the asymptotic theory in finite samples.

The rest of this paper proceeds as follows. Section \ref{sec:concept}
introduces the notation, model, and assumptions for identification. Section
\ref{sec:Methodology} presents the PSM estimator, including the matching procedure and estimating equations. In Section \ref{sec:Main-results}, we show the main asymptotic properties
of the PSM estimator. Section
\ref{sec:Resampling-variance-estimation} proposes a resampling-based
inference procedure. Section \ref{sec:Simulation-Study} uses simulation
to evaluate the finite-sample properties of the estimators. Section
\ref{sec:An-application} applies the new estimator to the IMS Health
Oncology electronic medical record data to evaluate the effects of
two treatments on treating non-small cell lung cancer. Section \ref{sec:Discussions}
concludes with potential extensions. The supporting information contains
the technical details, proofs and additional simulation results.

\section{Notation, Model, and Assumptions}\label{sec:concept}

\subsection{Potential outcomes and the causal PH model}

We use the potential outcomes framework under the Stable Unit Treatment
Value Assumption, and let $T^{(\omega)}$ and $C^{(\omega)}$ be the
potential values of the survival outcome and censoring indicator had
the individual received treatment $\omega$ ($\omega=0,1$). We assume
non-informative censoring under which $T^{(\omega)}\indep C^{(\omega)}$,
where $\indep$ means "independent of". This assumption is reasonable
if the censored times occur at the end of study follow-up $\tau$,
the so-called administrative censoring. 

We define
$
\lambda_{\omega}(t)=\lim_{\delta_t\rightarrow0}{\delta_t}^{-1}P\left(t\leq T^{(\omega)}<t+\delta_t\mid T^{(\omega)}\geq t\right)
$
as the causal hazard rate of failing at time $t$ for a population
of patients had they received treatment $\omega$. 
Adopting notation used by \citet{cox1972regression} in the potential outcomes framework,
we define $U^{(\omega)}=\min(T^{(\omega)},C^{(\omega)})$ as the time
to a clinical event or censoring, $\Delta^{(\omega)}=I(T^{(\omega)}\leq C^{(\omega)})$
the clinical event indicator, $N^{(\omega)}(t)=I(U^{(\omega)}\leq t,\Delta^{(\omega)}=1)$
the counting process of observed event and $Y^{(\omega)}(t)=I(U^{(\omega)}\geq t)$
the at-risk process \textcolor{black}{for a population of patients
had they received treatment $\omega$.} 
Following \citet{tchetgen2012parametrization}, we assume a causal PH
model.

\begin{definition*}[Causal PH Model]The structural
model for comparing treatment $\omega$ and treatment $0$ is\textcolor{black}{{}
\begin{equation}
\lambda_{\omega}(t)=\lambda_{0}(t)\exp\left(\beta_{0}\omega\right),\label{(5.1)}
\end{equation}
}where $\lambda_{0}(t)$ is the population hazard rate if all individuals
had treatment $\omega=0$. \end{definition*}

The parameter $\beta_0$ describes log of the relative
hazard of having a clinical event if all individuals received treatment
$\omega=1$ compared to if all individuals received treatment $\omega=0$.
Although focusing on the marginal hazard ratio is controversial in
the literature, it is still a
useful estimand that summarizes the overall average of the hazard
ratios over a certain time period \citep{hernan2000marginal}. This explains its wide acceptance
among clinical trial statisticians and regularity agencies such as
the U.S. Food and Drug Administration.

\subsection{Observed data and identification assumptions}

Let $W_{i}$ be the binary treatment, and let $T_{i}$ and $C_{i}$ be the times to a clinical event and censoring, respectively, for individual $i=1,\ldots,n$.
We define $U_{i}=\min(T_{i},C_{i})$ as the time to a clinical event
or censoring, $\Delta_{i}=I(T_{i}\leq C_{i})$ the clinical event
indicator, $N_{i}(t)=I(U_{i}\leq t,\Delta_{i}=1)$ the observed data
counting process, and $Y_{i}(t)=I(U_{i}\geq t)$ the at-risk process.
Suppose we observe a set of pre-treatment baseline covariates $X_{i}\in\R^{d}$. Let $e(X_{i})=P(W_{i}=1\mid X_{i})$ be the propensity score. 
The observed data are summarized
as $\{O_{i}=(X_{i},W_{i},U_{i},\Delta_{i}):i=1,\ldots,n\}$. We assume
that $\{O_{i}:i=1,\ldots,n\}$ are independent and identically distributed. 

We make the consistency assumption that links the observed data processes
with the potential outcome processes. In order to use the observed data to estimate
the parameters in Model (\ref{(5.1)}), we require the assumptions
of unconfoundedness and positivity \citep{robins2004optimal}.

\begin{assumption}[Consistency]\label{asump:consistency} $T_{i}=T_{i}^{(W_{i})}$
and $C_{i}=C_{i}^{(W_{i})}$, or equivalently \textcolor{black}{$N_{i}(t)=N_{i}^{(W_{i})}(t)$
and $Y_{i}(t)=Y_{i}^{(W_{i})}(t)$ for all $t$.}
\end{assumption}
\begin{assumption}[Unconfoundedness]\label{asp:NUC}$W_{i}\indep\{T_{i}^{(0)},T_{i}^{(1)}\}\mid X_{i}$.
\end{assumption}
\begin{assumption}[Positivity]\label{asp:positivity}
With probability $1$, $0<\underline{c}<e(X_{i})<\bar{c}<1$.
\end{assumption}

\section{The PSM estimator of the marginal hazard ratio} \label{sec:Methodology}

Although it is evident that paired matching (propensity score matching without replacement) is unbiased for the marginal hazard ratio over the treated population, this approach does not target estimation of $\beta_0$, the marginal hazard ratio over the population containing both treated and untreated individuals. Unlike paired matching, the PSM (with replacement) estimator permits estimation of $\beta_0$ and will be the focus of this article \citep{abadie2006large}. Previous simulation studies have provided evidence for the unbiasedness of PSM for estimating $\beta_0$ \citep{austin2020variance}. In this section, we describe the PSM estimator to set the stage for our main results.

For the units in the sample, only one of the potential counting processes, $N_{i}^{(0)}(t)$ and $N_{i}^{(1)}(t)$, is observed and the other is unobserved or missing. The same is true for the potential at-risk processes. The construction of the PSM estimator first involves imputing all the missing potential outcome processes. To do so, for each unit $i$, it finds the first $M$ closest units in the opposite treatment based on the Euclidean distance between the propensity scores. Note that $M$ is fixed and does not vary between units. Define $\mathcal{J}_{M,i}$ as the set of indices for the first $M$ matches for unit $i$. For all $i$ and $\omega$, the PSM estimator imputes the missing counting processes and at-risk processes respectively as
\begin{align}
\overline N_{i}^{*(\omega)}(t)  & =\begin{cases}
N_{i}(t)  & \text{if }W_{i}=\omega,\\
 {M}^{-1}\sum_{j\in\mathcal{J}_{M,i}}N_{j}(t)   & \text{if }W_{i}=1-\omega,
\end{cases}\label{eq:imputation-1}
\end{align}
and
\begin{align}
\overline Y_{i}^{*(\omega)}(t)  & =\begin{cases}
Y_{i}(t)  & \text{if }W_{i}=\omega,\\
 {M}^{-1}\sum_{j\in\mathcal{J}_{M,i}}Y_{j}(t)   & \text{if }W_{i}=1-\omega.
\end{cases}\label{eq:imputation-2}
\end{align}

Let $k_{i}=\sum_{l=1}^n I\left(i \in \mathcal{J}_{M, l}\right)$ denote the number of times unit $i$ is used as a match. Note that each unit can be used as a match more than once so $k_{i}$ can be larger than 1. In practice, the true propensity score is unknown. Following most of the empirical literature,
we assume that the propensity score follows a generalized linear model,
$e(X_{i}^{\T}\theta_0)$ \citep{imbens2015causal,li2016balancing,austin2017performance}. The matching procedure
can be carried out with the estimated propensity score $e(X_{i}^{\T}\widehat{\theta})$, where  $\widehat{\theta}$ is a consistent estimator of $\theta_0$. We now denote $k_{i}$ to be $k_{\widehat{\theta},i}$ to reflect its dependence on  $\widehat{\theta}$.

Once the missing potential outcome processes are imputed, PSM then fits a marginal structural model to the imputed dataset. Define $\Lambda_{0}(t)=\int_{0}^{t}\lambda_{0}(v)\de v$ as the cumulative
hazard function for $\omega=0$ at time $t$. The
estimating functions for $\Lambda_{0}(t)$, $t\geq0$ and $\beta_{0}$ are
\begin{eqnarray}
\sum_{i=1}^{n}\sum_{\omega=0}^{1}\left\{ \de\overline{N}_{i}^{*(\omega)}(t)-\de\Lambda_{0}(t)\exp(\beta\omega)\overline{Y}_{i}^{*(\omega)}(t)\right\} ,\label{eq:(13.1)}\\
\sum_{i=1}^{n}\sum_{\omega=0}^{1}\int_{0}^{\tau}\omega\left\{ \de\overline{N}_{i}^{*(\omega)}(t)-\de\Lambda_{0}(t)\exp(\beta\omega)\overline{Y}_{i}^{*(\omega)}(t)\right\} ,\label{eq:(13.2)}
\end{eqnarray}
where $\overline{N}_{i}^{*(\omega)}(t)$ and $\overline{Y}_{i}^{*(\omega)}(t)$
are  defined in (\ref{eq:imputation-1}) and (\ref{eq:imputation-2}).

We can equivalently write (\ref{eq:(13.1)})
and (\ref{eq:(13.2)}) using observed outcome processes respectively as 
\begin{eqnarray}
\sum_{i=1}^{n}\left\{ 1+\frac{k_{\widehat{\theta},i}}{M}\right\} \left\{ \de N_{i}(t)-\de\Lambda_{0}(t)\exp(\beta W_{i})Y_{i}(t)\right\} ,\label{eq:(16.1)}\\
\sum_{i=1}^{n}\left\{ 1+\frac{k_{\widehat{\theta},i}}{M}\right\} \int_{0}^{\tau}W_{i}\left\{ \de N_{i}(t)-\de\Lambda_{0}(t)\exp(\beta)Y_{i}(t)\right\} .\label{eq:(16.2)}
\end{eqnarray}
Setting (\ref{eq:(16.1)}) equal to zero, we can obtain the estimator
for $\de\Lambda_{0}(t)$ for fixed $\beta$ as 
\begin{equation}
\de\widehat{\Lambda}_{0}(t)=\frac{\sum_{i=1}^{n}\{1+k_{\widehat{\theta},i}/M\}\de N_{i}(t)}{\sum_{i=1}^{n}\{1+k_{\widehat{\theta},i}(W_{i})/M\}\exp(\beta W_{i})Y_{i}(t)}.\label{(17.1)}
\end{equation}
Substituting (\ref{(17.1)}) into (\ref{eq:(16.2)}), we obtain the
estimating equation to solve for $\beta$, 
\begin{equation}
G_{n}(\beta)=\sum_{i=1}^{n}\int_{0}^{\tau}\left\{ 1+\frac{k_{\widehat{\theta},i}}{M}\right\} \left\{ W_{i}-\widehat{Q}(\beta,t)\right\} \de N_{i}(t)=0,\label{eq:partialscore}
\end{equation}
where 
\begin{equation}
\widehat{Q}(\beta,t)=\frac{\sum_{j=1}^{n}\{1+k_{\widehat{\theta},i}/M\}W_{j}\exp(\beta W_{j})Y_{j}(t)}{\sum_{j=1}^{n}\{1+k_{\widehat{\theta},i}/M\}\exp(\beta W_{j})Y_{j}(t)}.\label{eq:(2.1)}
\end{equation}
Equation
(\ref{eq:partialscore}) is the partial score equation for
a Cox PH model with $W_{i}$ as the covariate and weights
$1+k_{\widehat{\theta},i}/M$. Thus, the PSM estimator $\widehat{\beta}$
can be calculated by standard software.

We conclude this section by summarizing the steps for calculating the PSM estimator for $\beta_{0}$ as follows. 
\begin{description}
\item [{{Step$\ $1.}}] 
Fit a propensity score model $e(X_{i}^{\T}\theta)$,
and obtain an estimate $\widehat{\theta}$. 
\item [{{Step$\ $2.}}] Based on the estimated propensity scores $e(X_{i}^{\T}\widehat{\theta})$, carry out the matching procedure described above. Record the number of times each unit is used as a match $k_{\widehat{\theta},i}$. 
\item [{{Step$\ $3.}}] Obtain the PSM estimator $\widehat{\beta}$ for $\beta_{0}$ by solving
(\ref{eq:partialscore}) using standard software, i.e. by fitting a Cox PH model to the observed data with covariate $W_i$ and weights $1+k_{\widehat{\theta},i}/M$. 
\end{description}


\section{Main Results}\label{sec:Main-results}
\subsection{Asymptotic results with known propensity scores}\label{subsec:Asymptotic-results-known}

In this section, we establish the asymptotic normality of the PSM estimator of the marginal HR assuming the propensity scores are known. The results in this subsection are applicable to rare situations when the propensity scores are known, and are useful for understanding the effect of estimating the propensity scores on the PSM estimator when compared to the results in the next subsection.

Under regularity conditions
in Assumption S1, we show that there exists a neighborhood $\mathcal{B}$ of $\beta_{0}$
and a function $Q(\beta_{0},t)$ such that for all $(\beta,t)\in\mathcal{B}\times[0,\tau]$,
$\widehat{Q}(\beta,t)\rightarrow_{p}Q(\beta,t),\text{ as }n\rightarrow\infty.
$ 
We then define
\begin{equation}
H_{i}(\omega)=\int_{0}^{\tau}\left\{ \omega-Q(\beta_0,t)\right\} \left\{ \de N_{i}^{(\omega)}(t)-\de\Lambda_{0}(t)\exp(\beta_{0}\omega)Y_{i}^{(\omega)}(t)\right\} ,\label{eq:def H(w)}
\end{equation}
$\mu_{H}(\omega,X)=E\{H(\omega)\mid X\},$ and $\sigma_{H}^{2}(\omega,X)=\var\{H(\omega)\mid X\}.$ 

Because $\widehat{\beta}$ is the solution to the estimating equation $G_{n}(\beta)=0$ in
(\ref{eq:partialscore}), the key step is to characterize the asymptotic
properties of  $G_{n}(\beta_{0})$. With a known $\theta_{0}$, 
we can show that
\begin{equation}
n^{-1/2}G_{n}(\beta_{0})=n^{-1/2}\sum_{i=1}^{n}\left\{ 1+\frac{k_{\theta_{0},i}}{M}\right\} H_{i}(W_{i})+o_{p}(1).\label{eq:linear S}
\end{equation}
Based on (\ref{eq:linear S}) and the M estimation theory,  
we derive the asymptotic results for $\widehat{\beta}$ as follows.

\begin{theorem} \label{Thm1} 
Under Assumptions \ref{asump:consistency}--\ref{asp:positivity}
and the regularity conditions in Assumption S1 presented in the supplementary
material, with the known propensity score, 
\[
n^{1/2}(\widehat{\beta}-\beta_{0})\rightarrow\N(0,V_{1}),
\]
as $n\rightarrow\infty$, where $V_{1}=\{A(\beta_{0})\}^{-1}V_{G}\{A(\beta_{0})\}^{-1}$,
\begin{align}
A(\beta_{0}) & =E\Bigg(\int_{0}^{\tau}\left[\frac{E\left\{ \exp(\beta_{0})Y^{(1)}(t)\right\} }{E\left\{ Y^{(0)}(t)\right\} +E\left\{ \exp(\beta_{0})Y^{(1)}(t)\right\} }-1\right]\nonumber \\
 & \qquad\qquad\qquad\quad\times\frac{E\left\{ \exp(\beta_{0})Y^{(1)}(t)\right\} }{E\left\{ Y^{(0)}(t)\right\} +E\left\{ \exp(\beta_{0})Y^{(1)}(t)\right\} }\sum_{\omega=0}^{1}\de N^{(\omega)}(t)\Bigg).
\end{align}
and,
\begin{align}
V_{G}=\sum_{\omega=0}^{1}E\left[\sigma_{H}^{2}\{\omega,e(X)\}\left\{ \frac{2M+1}{2Mp(\omega\mid X)}-\frac{p(\omega\mid X)}{2M}\right\} \right]+E\left(\left[\mu_{H}\{0,e(X)\}+\mu_{H}\{1,e(X)\}\right]^{2}\right),\label{eq:VS}
\end{align}
$p(\omega\mid X)=\pr(W=\omega\mid X)$.
\end{theorem}


    Recall that the estimating
    function of the PSM estimator is $n^{-1/2}G_{n}(\beta_{0})$ as in (\ref{eq:linear S}), which
    targets estimating $E\{\sum_{\omega=0}^{1}H_{i}(\omega)\}$. Theorem \ref{Thm1} shows that this estimating function has
    the asymptotic variance $V_{G}$ as in (\ref{eq:VS}). From the standard
    semiparametric estimation literature,
    the semiparametric efficiency bound for the target parameter $E\{\sum_{\omega=0}^{1}H_{i}(\omega)\}$
    is $\sum_{\omega=0}^{1}E\{\sigma_{H}^{2}(\omega,X)/p(\omega\mid X)\}+E[\{\mu_{H}(0,X)+\mu_{H}(1,X)\}^{2}]$ \citep{bang2005doubly}.
    Thus, the PSM estimator does not attain the semiparametric efficiency bound.
    The asymptotic variance $V_{G}$ in (\ref{eq:VS}) becomes closer to the efficiency
    bound as the number of matches $M$ gets larger and $e(X)$ can explain all the variation
    of $H_{i}(\omega)$ given $X$, i.e., $\mu_{H}(\omega,X)=\mu_{H}\{\omega,e(X)\}$.

\subsection{Asymptotic results with estimated propensity scores}\label{subsec:Asymptotic-results-estimated}

We now study the asymptotic properties of the PSM estimator with the estimated propensity score. The
technique we will use is based on \citep{andreou2012alternative}.
The main idea is to apply
Le Cam's third lemma 
to locally asymptotically normal models \citep{le1990asymptotics}. Let $P^{\theta}$ be the distribution of $\left\{\left(A_i, X_i, U_i, \Delta_i \right): i=1, \ldots, n\right\}$ indexed by the parameter $\theta$ from the propensity score model. Let $\theta_{n}$ be contiguous to $\theta_{0}$, and $P^{\theta_{n}}$ be the distribution indexed by the local parameter $\theta_{n}$.
Under $P^{\theta_{n}}$, denote the true parameter value as $\beta_{0}(\theta_{n})$,
the PSM estimator based on the true parameter
$\theta_{n}$ as $\widehat{\beta}(\theta_{n})$, and the log likelihood
of $P^{\theta_{0}}$ with respect to $P^{\theta_{n}}$ as $\Lambda_{n}(\theta_{0}\mid\theta_{n})$.
Assume that {\it under $P^{\theta_{n}}$},
\begin{equation}
\left(n^{1/2}\{\widehat{\beta}(\theta_{n})-\beta_{0}(\theta_{n})\},n^{1/2}(\widehat{\theta}-\theta_{n}),\Lambda_{n}(\theta_{0}\mid\theta_{n})\right)^{\T}\label{eq:LAN}
\end{equation}
has a limiting normal distribution.  Le Cam\textquoteright s third
lemma states that {\it under $P^{\theta_{0}}$}, $n^{1/2}\{\widehat{\beta}(\theta_{n})-\beta_{0}(\theta_{n})\}$
has a limiting normal distribution. Then heuristically by replacing $\theta_{n}$
with $\widehat{\theta}$, one can then approximate the asymptotic distribution
of $n^{1/2}(\widehat{\beta}-\beta_{0})$ as shown in Theorem \ref{Thm2}. 

\begin{theorem}\label{Thm2}

Under Assumptions \ref{asump:consistency}--\ref{asp:positivity}
and the regularity conditions in Assumption S1 presented in the supplementary
material, with a correctly specified propensity score model, 
\[
n^{1/2}(\widehat{\beta}-\beta_{0})\rightarrow\N(0,V_{2}),
\]
as $n\rightarrow\infty$, $V_{2}=\{A(\beta_{0})\}^{-1}\widetilde{V}_{G}\{A(\beta_{0})\}^{-1}$,
$\widetilde{V}_{G}=V_{G}-c^{\T}\mathcal{I}_{\theta_{0}}^{-1}c,$ where
$V_{G}$ is defined in Theorem \ref{Thm1}, $\mathcal{I}_{\theta_{0}}$
is the Fisher information of $\theta_{0}$,  $\dot{e}(X^{\T}\theta_{0})=\partial e(X^{\T}\theta_{0})/\partial\theta$, and
\begin{equation}
c=E\left\{ \left[\frac{\cov\left\{ X,\mu_{H}(1,X)\mid e(X)\right\} }{e(X)}+\frac{\cov\left\{ X,\mu_{H}(0,X)\mid e(X)\right\} }{1-e(X)}\right]\dot{e}(X^{\T}\theta_{0})\right\} .\label{eq:c}
\end{equation}

\end{theorem}
By comparing $V_2$ and $V_1$, the change in asymptotic variance after adjusting for estimating propensity score function is $-\{A(\beta_{0})\}^{-1}c^{\T}\mathcal{I}_{\theta_{0}}^{-1}c\{A(\beta_{0})\}^{-1}$, which can either be negative or zero. This reduction is a result of utilizing the available treatment assignment information in the data, which improves the efficiency of the propensity score matching estimator.  Therefore, Theorem 2 shows that matching based on the estimated propensity score generally improves the efficiency of the matching estimator compared to matching based on the true propensity score even if it is known. This phenomenon
is in line with that in the setting with a continuous outcome \citep{abadie2016matching}.
Theorem 2 motivates a variance estimator based on an approximation of the asymptotic variance formula (see supplementary information).

\section{Resampling-based inference}\label{sec:Resampling-variance-estimation}

We propose resampling variance estimation that
has a better finite-sample performance. 
\citet{abadie2008failure} demonstrated that the nonparametric bootstrap
is invalid for the matching estimator based on matching with replacement
and with a fixed number of matches. \citet{otsu2017bootstrap} suggested
resampling the matching estimator directly based on its martingale
residual terms, which only works for matching based on the full vector
of covariates. In order to reflect the uncertainty in the estimation
of the propensity score, \citet{adusumilli2022bootstrap} proposed
a hybrid bootstrap procedure by re-assigning new values for the treatments and
resampling the martingale residuals under both treatment conditions.
This necessitates imputation of the unobserved martingale residuals
under the opposite treatment.
Extending this procedure to survival outcomes, 
we define the martingale residuals as 
\begin{eqnarray*}
\widehat{r}_{1i}(\theta) & = & \widehat{\mu}_{H}\left\{ 0,e(X^{\T}\theta)\right\} +\widehat{\mu}_{H}\left\{ 1,e(X^{\T}\theta)\right\} -0,\\
\widehat{r}_{2i}(\omega,\theta) & = & H_{i}(\omega)-\widehat{\mu}_{H}\left\{ \omega,e(X^{\T}\theta)\right\} ,%
\end{eqnarray*}
where $\widehat{\mu}_{H}(\omega,e(X_{i}^{\T}\theta))$ is obtained by a nonparametric
regression estimation of $H_{i}$ on $e(X_{i}^{\T}\theta)$
among individuals with $W_{i}=\omega$. 
For individual $i$, define the secondary nearest
neighbor matching function as $m(\omega,X_{i})$; if $W_{i}=\omega,$
$m(\omega,X_{i})=i$, otherwise, $m(\omega,X_{i})$ returns the index
of the nearest neighbor in the opposite treatment group, where nearest neighbor is determined based on the full $X$ rather than on the propensity score. 
The pair of potential residuals for individual $i$ are defined as
\begin{equation}
\widehat{r}_{2i}^{*}(0,\theta)=\begin{cases}
\widehat{r}_{2i}(0,\theta) & \text{if }W_{i}=0,\\
\widehat{r}_{2m(0,X_{i})}(0,\theta) & \text{if }W_{i}=1,
\end{cases}\ \ \widehat{r}_{2i}^{*}(1,\theta)=\begin{cases}
\widehat{r}_{2m(0,X_{i})}(1,\theta) & \text{if }W_{i}=0,\\
\widehat{r}_{2i}(1,\theta) & \text{if }W_{i}=1.
\end{cases}\label{eq:imputation2}
\end{equation}
We propose a double-resampling
procedure as follows. 
\begin{description}
\item [{{Step$\ $0.}}] 
Complete Steps 1-3 for obtaining the PSM estimator. For each individual $i$, compute the secondary nearest neighbor matching function $m(\omega,X_{i})$ for $\omega=0,1$. Estimate the matching function $k_{\theta_{0},i}(\omega)$
for $\omega=0,1$. For $\omega=W_{i}$, let $\widehat{k}_{i}(\omega)=k_{\widehat{\theta},i}$;
for $\omega=1-W_{i}$, we use the following imputation strategy: create
$q_N$ quantile partitions based on $e(X_{i}^{\T}\widehat{\theta})$, identify the
quantile partition individual $i$ falls, randomly sample one, say $l$, from
that partition with $W_{j}=\omega$, and let $\widehat{k}_{i}(\omega)=k_{\widehat{\theta},l}$. 
\item [{{Step$\ $1.}}] Generate the bootstrap treatment, $W_{i}^{*},i=1,\ldots,n$,
from $0$ with probability $1-e(X_{i}^{\T}\widehat{\theta})$, and $1$ with probability $e(X_{i}^{\T}\widehat{\theta})$.
\item [{{Step$\ $2.}}] Based on $(W_{i}^{*},X_{i})_{i=1}^{n}$, re-fit
the propensity score model and obtain the replicate $\widehat{\theta}^{*}$. 
\item [{{Step$\ $3.}}] Generate a sequence of independent and identically
distributed random variables $(u_{1}^{*},\ldots,u_{n}^{*})$ from
$u_{i}^{*}\sim\N(0,1)$. 
\item [{{Step$\ $4.}}] Define the new bootstrap residuals as 
$
\widehat{r}_{H,i}^{*}(\theta)=\widehat{r}_{1i}(\theta)+\{1+k_{\widehat{\theta},i}(W_{i}^{*})/M\}\widehat{r}_{2i}(W_{i}^{*},\theta),
$
whose expectation over the probability distribution implied by Step
1, conditional on the original data, is 
$\widehat{R}_{H}^{*}(\theta)={n}^{-1}\sum_{j=1}^{n}[\widehat{r}_{1j}(\theta)  +e(X^{\T}\theta)\{1+k_{\widehat{\theta},i}(1)/M\}\widehat{r}_{2j}(1,\theta)
  +\{1-e(X^{\T}\theta)\}\{1+k_{\widehat{\theta},i}(0)/M\}\widehat{r}_{2j}(0,\theta)]$, 
where $k_{\widehat{\theta},i}(\omega)$ is imputed at Step 0. Re-center
the bootstrap residuals and compute the replicate of $G_{n}(\beta_{0})$
as 
 $G_{n}^{*}=\sum_{i=1}^{n}\left\{ \widehat{r}_{H,i}^{*}(\widehat{\theta}^{*})-\widehat{R}_{H}^{*}(\widehat{\theta}^{*})\right\} u_{i}^{*}.$
 
\item [{{Step$\ $5.}}] Repeat Steps 1--4 for a large number $B$ times
and denote the $b$th replicate of $G_{n}(\beta_{0})$ as $G_{n}^{*(b)}$.
Construct the $100(1-\alpha)\%$ percentile bootstrap confidence interval
for $G_{n}(\beta_{0})$ as $\left(G_{n(\alpha/2)}^{*},G_{n(1-\alpha/2)}^{*}\right)$,
where $G_{nq}^{*}$ is the $q$th percentile of $\{G_{n}^{*(1)},\ldots,G_{n}^{*(B)}\}$. 
\end{description}

We construct the $100(1-\alpha)\%$ confidence interval for $\beta_{0}$
as $\big\{\de G_{n}(\widehat{\beta})/\de\beta\big\}^{-1}(G_{n(\alpha/2)}^{*},G_{n(1-\alpha/2)}^{*})$,
which has the nominal coverage level asymptotically. The double-resampling
procedure is parallel to the weighted bootstrap procedure of \citet{adusumilli2022bootstrap}. 
Thus, the validity of the double-resampling procedure is a straightforward
adaptation of the proofs in the work of \citet{adusumilli2022bootstrap}
to our setting under regularity conditions given in Assumption S2. 
The interested reader is encouraged to consult
the original article by \citet{adusumilli2022bootstrap} for 
details. 

\begin{remark}We provide some discussions to Step 0. Note that
the secondary matching procedure matches on the full set of covariates
rather than on the propensity scores. Doing so preserves the conditional
covariance in (\ref{eq:c}) between $X$ and the error terms $\widehat{r}_{2i}$,
given the propensity scores. This term reflects the improvement when
using the estimated propensity score. 
 See Adusumilli.\cite{adusumilli2022bootstrap} We impute $k_{i}(\omega)$
by drawing a value from the empirical distribution of $k_{i}(W_{i})$
in the propensity score strata for $W_{i}=\omega$ to re-create the
distribution of $k_{i}(W_{i})$. The number of the propensity score
strata $q_N$ is required to go to infinity as the sample size increases.
In finite samples, we suggest using quintiles ($q_N=5$) and recommend conducting
sensitive analysis with different choices of $q_N$. Our simulation study shows
that the performance of the proposed double-sampling procedure is
not sensitive to this choice.


\end{remark}

\section{Simulation study}\label{sec:Simulation-Study}

\subsection{Overview}

We conduct a simulation study to compare the performance of PSM estimator to existing approaches for estimating the log-marginal hazard ratio and to assess the performance of the proposed resampling-based variance estimator. Motivated by previous findings, we consider data generating mechanisms that simulate a varying level of covariate overlap under the marginal proportional hazard assumption. This simulation design attempts to highlight the differences between the PSM estimator and other existing approaches, by investigating the following hypotheses:

\begin{itemize}
\item[{(1)}] When overlap is poor or when the propensity score model is misspecified, the PSM estimator is anticipated to outperform the IPW and AIPW estimators because PSM is more robust to extreme values of the propensity scores \citep{kang2007demystifying,waernbaum2012model,busso2014new}. See supplementary information for a discussion comparing PSM and AIPW from a theoretical and practical point of view.

\item[{(2)}] Matching on high-dimensional covariates will result in noticeable bias because the search for close matches becomes increasingly difficult in higher dimensions \citep{abadie2006large,yang2016propensity,zhao2022double,zhao2023advances}.

\item[{(3)}] The proposed resampling-based variance estimator is expected to outperform other existing approaches \citep{adusumilli2022bootstrap}.

\end{itemize}

Bias, empirical variance, the average of variance estimates, and the coverage rate for nominal 95\% confidence intervals are obtained. The performance of the point estimators are evaluated using bias and empirical variance. Any deviation of the averaged variance estimates from the empirical variance reflects bias in variance estimation. The coverage rate for nominal 95\% confidence intervals of the estimated log-marginal hazard ratios is used to measure the overall performance of the estimation procedures. 

\subsection{Data generating mechanism}

We consider a sample size of $n=1000$ throughout. For individuals $i=1,...,n$, we simulate the following data: 
\begin{enumerate}[label=\textbf{Step \arabic*:},leftmargin=*]
\item vector of baseline covariates $X=(X_{1},\cdots,X_{6})^\T$, each independently generated from an exponential distribution with mean 1.
\item indicator of assignment to treatment $W$ generated from the propensity score model: $\text{logit} \; P(W=1\mid X)=\theta_{0}+X^{\T}\theta_{1}$, where we let $(\theta_{0},\theta_{1}^{\T})$ be $(-3.0,0.5\times\mathbf{1}_{6\times 1}^\T)$, $(-4.5,0.75\times\mathbf{1}_{6\times 1}^\T )$, and $(-5.0,\mathbf{1}_{6\times 1}^\T )$ to represent weak, moderate, and strong levels of confounding, leading to strong, moderate, and weak covariate overlap. The distribution of the true propensity scores is shown in Figure 1. We also examine a scenario where there is no confounding, i.e. each subject has a true propensity score of 0.5 (see supplementary information). 

\item counterfactual outcome under control $T^{(0)}$ generated from an exponential distribution with survival function $S^{(0)}(t)=\exp(-\lambda_{0}t)$; counterfactual outcome under treatment $T^{(1)}$ generated by first drawing $u$ from a uniform distribution with support $[0,1]$ and then solving $$\left\{ \prod_{k=1}^{6}\left(1-\eta_{k}t\right)\right\} \exp\left[\left\{ X^{\T}\eta-\lambda_{0}\exp(\beta_{0})\right\} t\right]-1+u=0$$
for $t$, where $\eta_{1}=\cdots=\eta_{6}=-2$ and $\eta=\left(\eta_1,\cdots,\eta_6\right)^{\T}$. We consider three choices for the estimand $\beta_0$ and let $\lambda_{0}=6$ for $\beta_0=0$ and $\beta_0=0.5$, and $\lambda_{0}=15$ for $\beta_0=-0.5$. This method for simulating the counterfactual outcomes has been used to guarantee the PH assumption \citep{wang2022instrumental}. See supplementary information for justification.

\item actual true event time is calculated as $T=(1-W)T^{(0)}+WT^{(1)}$. 

\item event time $\Delta=I(T<C)$ and event indicator $U=\min(C,T)$, where $C$ is generated from a uniform distribution with support $[0,a]$ and $a$ is set so that between $20\%$ to $30\%$ individuals are censored.
\end{enumerate}

In sum, we allow the following two factors to vary in our simulation designs: covariate overlap (weak, moderate, strong, and very strong) and the true log of marginal hazard ratio (0, -0.5, 0.5). This gives rise to a total of 12 scenarios. For each scenario, we simulate 1000 datasets.

\subsection{Methods}

We compare the following estimators of the log-marginal hazard ratio:
\begin{itemize}
\item[{(i)}] 
$\widehat{\beta}_{\nai}$, the unadjusted estimator obtained by fitting the Cox PH model with the treatment status as the only covariate, which serves as a benchmark for demonstrating the level of confounding bias; 
\item[{(ii)}]  
$\widehat{\beta}_{\ipw}$,
the IPW estimator obtained by fitting a weighted Cox PH model  with the
treatment status as the only covariate, where each observation is weighted by the inverse of
the probability of receiving the actual treatment; 
\item[{(iii)}]  
$\widehat{\beta}_{{\rm aipw}}$, the AIPW estimator, where the working outcome model is the Cox PH model with all baseline covariates and
treatment status as covariates \citep{tchetgen2012parametrization};
\item[{(iv)}]  $\widehat{\beta}_{{\rm m.x}}$, the estimator based on matching on all covariates, which works the same way as the PSM estimator, except during the matching stage, nearest neighbors are determined based on the Euclidean distance between vectors of covariates rather than between propensity scores; the number of matches is set to $M=1$;

\item[{(v)}]  $\widehat{\beta}_{\psm,0}$, the PSM estimator 
based on the true propensity score; the number of matches is set to $M=1$;
\item[{(vi)}]  $\widehat{\beta}_{\psm}$, the PSM estimator 
based on the estimated propensity score; two versions of the PSM estimator corresponding to two different choices for the number of matches ($M=1$ and $M=5$) are considered.
\end{itemize}

For variance estimation of $\widehat{\beta}_{\nai}$, $\widehat{\beta}_{\ipw}$, and $\widehat{\beta}_{{\rm m.x}}$, we use the robust output from the
standard software. 
The nonparametric bootstrap is used for estimating the variance of $\widehat{\beta}_{\aipw}$. For $\widehat{\beta}_{\psm,0}$, we use the proposed asymptotic variance estimator without the adjustment term. 
We compare four variance estimators for $\widehat{\beta}_{\psm}$: (i) the robust output from the standard software, (ii) the proposed asymptotic variance estimator (iii)  the
nonparametric bootstrap and (iv) the proposed double-resampling procedure with $q_N=5$.

\subsection{Sensitivity analysis}

For each simulation scenario above, we examine the effect of misspecifying the propensity score model on the three propensity score based approaches ($\widehat{\beta}_{\ipw}$, $\widehat{\beta}_{\rm aipw}$, and $\widehat{\beta}_{\psm}$). We also investigate the sensitivity of the double-resampling approach to a different choice for the number of strata ($q_{N}=10$). 

\subsection{Simulation study results}

Figure \ref{fig:Simulation-results} and Table \ref{tab:Simulation-results} show the results when true $\beta_{0}=0$ and the propensity score model is correctly specified.
The results for the misspecified propensity score model and results for $\beta_{0}=-0.5$ and $\beta_{0}=0.5$ are presented
in the supporting information. The unadjusted estimator $\widehat{\beta}_{\nai}$
has a severe bias and thus barely captures $\beta_{0}$, which becomes
worse as the covariate overlap between the treatment
groups becomes weaker. The matching estimator based on all covariates $\widehat{\beta}_{\xm}$ has the second-largest bias following the naive estimator, $\widehat{\beta}_{\nai}$, confirming the theoretical result in \citet{abadie2006large} that matching on more than one covariate may lead to a biased matching estimator. 
The IPW estimator $\widehat{\beta}_{\ipw}$ and the AIPW estimator $\widehat{\beta}_{\aipw}$ 
greatly reduce the bias; however, as shown in Figure \ref{fig:Simulation-results}, they are unstable when the two treatment groups have weak overlap in propensity scores. Moreover,  $\widehat{\beta}_{\aipw}$ proposed by \citet{tchetgen2012parametrization} 
is supposed to be doubly robust in the sense that its consistency relies only on the correct specification of either the propensity score model or a working outcome mean model given the covariates. However, in practice, specifying a correct outcome mean model that is compatible with the marginal structural model is difficult. The posited Cox regression model in the AIPW estimator is a convenient choice but is misspecified. Thus, the AIPW estimator is no longer doubly robust.
For all investigated simulation scenarios, $\widehat{\beta}_{\psm}$ can significantly reduce the bias and are more stable when the two treatment groups have weak overlap in propensity scores compared to $\widehat{\beta}_{\ipw}$ and $\widehat{\beta}_{\aipw}$. The PSM estimator with $M=5$ provides a smaller variance than the counterpart with $M=1$. This is because increasing the number of matching can improve the efficiency of the PSM estimator. As a trade-off, when the sample size is small, a large number of matching could lead to a slight increase in bias.  Moreover, the PSM estimators possess coverage rates around 95\% using the proposed asymptotic variance estimation and double-resampling. When the two treatment groups have weak overlap, it is noticed that the asymptotic variance estimator tends to underestimate the variance and may result in negative values. In contrast, the double-resampling method recovers the 95\% confidence level better than the asymptotic method for such cases.  Comparing to the proposed inference approach, using the standard software's robust method and naive bootstrap method always overestimate the variances of $\widehat{\beta}_{\psm,0}$ and $\widehat{\beta}_{\mathrm{psm}}$.
Thus, our  proposed variance estimation approach is apparently
beneficial for making a reliable inference.

\begin{table}[h]
\caption{\label{tab:Simulation-results}Simulation results: bias ($\times10^{2}$)
and variance ($\times10^{3}$) of the point estimator of $\beta_{0}$,
coverage ($\%$) of $95\%$ confidence intervals based on $1,000$
Monte Carlo samples with true $\beta_{0}=0$}
\resizebox{\textwidth}{!}{%
\begin{tabular}{llllllllllllll}
\hline
\multicolumn{1}{c}{}                                                                 &            & Bias & Var & VE   & CR     & Bias & Var  & VE   & CR     & Bias & Var  & VE   & CR     \\
\multicolumn{1}{c}{Level of covariate overlap}                                                &            & \multicolumn{4}{l}{Strong}   & \multicolumn{4}{l}{Medium}  & \multicolumn{4}{l}{Weak}  \\ \hline
$\widehat{\beta}_{\nai}$                                                                  &            & 54.6 & 4.9 & 5.1  & 0.0\%  & 70.6 & 5.2  & 5.3  & 0.0\%  & 80.9 & 5.2  & 5.7  & 0.0\%  \\
$\widehat{\beta}_{\ipw}$                                                                  &            & 0.3  & 6.1 & 8.3  & 97.8\% & 1.0  & 12.5 & 14.6 & 95.6\% & 4.1  & 22.2 & 18.9 & 91.2\% \\
$\widehat{\beta}_{\aipw}$                                                                   &            & 0.1  & 5.3 & 6.0  & 95.9\% & 0.2  & 9.1  & 12.0 & 96.1\% & 2.3  & 31.3 & 40.6 & 96.5\% \\
$\widehat{\beta}_{\xm}$                                                                   &            & 6.9  & 5.9 & 8.8  & 92.7\% & 11.4 & 8.2  & 11.8 & 84.0\% & 17.7 & 8.2  & 11.7 & 64.4\% \\
$\widehat{\beta}_{\psm,0}$                                                                &            & 0.5  & 6.6 & 10.9 & 98.5\% & 1.6  & 10.6 & 17.9 & 97.4\% & 3.7  & 16.1 & 25.4 & 96.1\% \\
\multirow{4}{*}{\begin{tabular}[c]{@{}l@{}}$\widehat{\beta}_{\psm}$\\ (M=1)\end{tabular}} & software   & 0.7  & 6.3 & 11.1 & 98.9\% & 1.3  & 10.5 & 18.3 & 97.9\% & 3.9  & 17.2 & 26.0 & 95.8\% \\
                                                                                     & asymp      & 0.7  & 6.3 & 6.5  & 95.1\% & 1.3  & 10.5 & 10.6 & 94.2\% & 3.9  & 17.2 & 16.4 & 91.3\% \\
                                                                                     & naiveboot  & 0.7  & 6.3 & 9.8  & 98.8\% & 1.3  & 10.5 & 17.0 & 98.7\% & 3.9  & 17.2 & 25.1 & 96.3\% \\
                                                                                     & double-rsp & 0.7  & 6.3 & 8.2  & 97.7\% & 1.3  & 10.5 & 13.6 & 97.5\% & 3.9  & 17.2 & 24.1 & 95.4\% \\
\multirow{4}{*}{\begin{tabular}[c]{@{}l@{}}$\widehat{\beta}_{\psm}$\\ (M=5)\end{tabular}} & software   & 1.5  & 5.0 & 8.1  & 99.0\% & 3.2  & 8.0  & 12.3 & 96.6\% & 5.1  & 10.1 & 16.3 & 97.0\% \\
                                                                                     & asymp      & 1.5  & 5.0 & 5.0  & 94.2\% & 3.2  & 8.0  & 7.3  & 91.9\% & 5.1  & 10.1 & 9.2  & 90.6\% \\
                                                                                     & naiveboot  & 1.5  & 5.0 & 7.6  & 99.0\% & 3.2  & 8.0  & 12.2 & 96.8\% & 5.1  & 10.1 & 16.6 & 97.5\% \\
                                                                                     & double-rsp & 1.5  & 5.0 & 6.1  & 96.1\% & 3.2  & 8.0  & 9.5  & 95.2\% & 5.1  & 10.1 & 13.9 & 95.8\% \\ \hline
\end{tabular}} Note: "Var" is the variance of point estimates of $\beta_{0}$
across 1000 simulated datasets; "VE" is the average variance estimation
for the point estimators over simulations, thus VE minus Var reflects
the bias in estimated variance; "CR" is the empirical coverage rate
of $95\%$ confidence intervals. The proposed PSM estimator $\widehat{\beta}_{\psm}$ is calculated with number of matches $M=1$ and 5. Four types of variance estimates for $\widehat{\beta}_{\psm}$ were compared: "software", output from
the standard software; "asymp", the proposed asymptotic variance
estimation; "naiveboot", the naive nonparametric bootstrap; "double-rsp",
the proposed double-resampling method. 
\end{table}

In the extended simulations (see the supporting information), we conduct a sensitivity analysis on a different choice of the number of strata, $q_{N}$, for the double-resampling approach. The results show that the estimated variances are similar with a difference less than $10^{-3}$ across all scenarios, indicating that the variance estimator is insensitive to the choice of $q_{N}$. Moreover, when the propensity score model is misspecified,  $\widehat{\beta}_{\ipw}$, $\widehat{\beta}_{\aipw}$ and $\widehat{\beta}_{\psm}$ become biased. Nonetheless,  $\widehat{\beta}_{\psm}$ is still more robust than $\widehat{\beta}_{\ipw}$ and $\widehat{\beta}_{\aipw}$. The AIPW estimator performs the worst among all estimators when the level of covariate overlap is medium and weak, consistent with the results in \citet{kang2007demystifying} that the bias and variance of AIPW estimator can increase dramatically when both the propensity score and outcome models are misspecified. For the scenario where there is no confounding, while $\widehat{\beta}_{\psm}$ with $M=5$ results in similar bias than $\widehat{\beta}_{\ipw}$ and $\widehat{\beta}_{\aipw}$, the PSM estimators are generally less efficient than the weighting approaches. Among all the variance estimators of $\widehat{\beta}_{\psm}$, the asymptotic variance estimator shows the smallest bias.

\begin{figure}
    \begin{center}
           \includegraphics[width=5in]{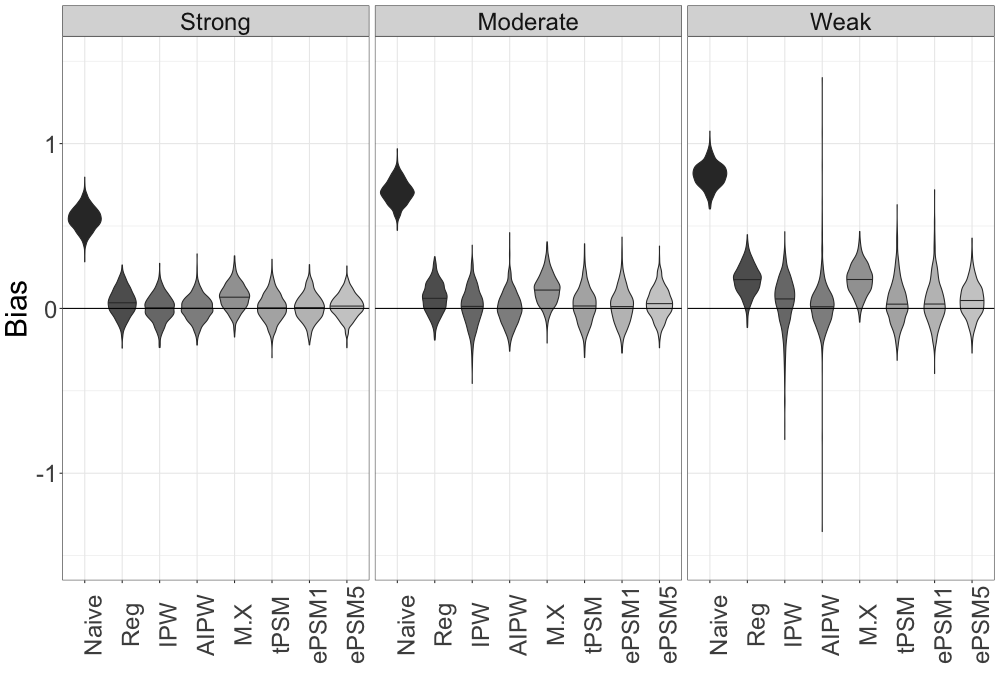}
    \caption{\label{fig:Simulation-results} Simulation results of various point estimators of $\beta_{0}$ from 1000 Monte Carlo simulated datasets. $\mathtt{Naive}$ is the $\widehat{\beta}_{\nai}$ estimator; $\mathtt{IPW}$ is the $\widehat{\beta}_{\ipw}$ estimator; $\mathtt{AIPW}$ is the $\widehat{\beta}_{\aipw}$ estimator; $\mathtt{M.X}$ is the $\widehat{\beta}_{{\rm m.x}}$ estimator; $\mathtt{tPSM}$ is the $\widehat{\beta}_{\psm,0}$ estimator; $\mathtt{ePSM1}$ is the $\widehat{\beta}_{\psm}$ estimator with $M=1$; $\mathtt{ePSM5}$ is the $\widehat{\beta}_{\psm}$ estimator with $M=5$; some AIPW point estimations with absolute value greater than 1.5 are excluded from the plot in the scenario where the level of covariate overlap is strong.} 
    \end{center}
\end{figure}

\section{An application}\label{sec:An-application}

Non-small cell lung cancer (NSCLC) is the most commonly diagnosed
lung cancer; typically, around half the NSCLC patients receiving chemotherapy
will receive additional treatment in the post-progression setting,
i.e., second-line treatment setting, where "carboplatin + paclitaxel"
and "erlotinib" are two historically commonly used treatments;
see \citet{cui2018application}. In this section, we use the IMS Health
Oncology electronic medical record (EMR) data to conduct a comparative
effectiveness analysis of the two treatments in the second-line setting
with the proposed causal inference of hazard ratio based on PSM estimator and other existing estimators mentioned in
the previous section.

The EMR data is deidentified observational patient-level clinical
data with demographic and baseline clinical characteristics collected
from medium and large community-based oncology practices across 50
states of the USA. The dataset used contains a retrospective cohort
of 10,634 eligible patients at least 18 years old who received at
least two lines of therapy, from 1 January 2007 to 31 December 2014;
see \citet{cui2018application} for details. 

Overall survival was defined as the time from the start date of second-line therapy to the death date. The death date of patients was assumed to be the last visit if a \textit{sufficient period} had elapsed between the last visit and the end of the EMR data; patients with the time between last visit and the end of dataset shorter than a \textit{sufficient period} were censored at the date of the last visit. Here we define a \textit{sufficient period} as at least twice the average visit interval in the 3 months prior to last visit, given that a patients normally have multiple visits to the clinic. Patients alive at the end of the time period were censored at the end date of the dataset. Missing data were classified into
its own category for each categorical variable; see \citet{de2018transitional}.
Among the eligible patients in the dataset, 1241 patients were treated
with ``carboplatin + paclitaxel", and 895 patients received single-agent
``erlotinib" as second-line therapy, while the remaining patients received one of three other therapies. For illustration of our approach, we subsetted the data to contain only patients receiving ``carboplatin + paclitaxel" or ``erlotinib," to compare the treatment effects of the two therapies.

The propensity scores were estimated using a logistic regression model
with predictors: age at the initiation of second-line therapy, gender,
race, region, disease stage at initial diagnosis, Eastern Cooperative
Oncology Group performance status score at initiation of second-line
therapy, facility types of academic or community cancer center, year
of index diagnosis, and days from index diagnosis to initiation of
second-line therapy; see \citet{cui2018application} for details.

Table \ref{tab:HR} shows estimated log hazard ratio $\widehat{\beta}$,
point estimates and 95\% confidence intervals for the hazard ratio
of "carboplatin + paclitaxel" to "erlotinib" based on the unadjusted
estimator $\widehat{\beta}_{\nai}$, the IPW estimator $\widehat{\beta}_{\text{ipw}}$, and matching estimator based on the full set of covariates $\widehat{\beta}_{\xm}$,  with the robust variance estimation from the standard software for constructing Wald confidence intervals. We also include the AIPW estimator proposed by \citet{tchetgen2012parametrization} where the working outcome model is the Cox PH model based on the same set of covariates as the propensity score model and treatment indicator. The variance of $\widehat{\beta}_{\aipw}$ is estimated by a naive bootstrap procedure. For the proposed inference based on the PSM estimator $\widehat{\beta}_{\psm}$ with the number of matches $M=1$ and 5, we constructed both Wald confidence intervals based on the empirical
asymptotic variance estimator and bootstrap percentile confidence
intervals with the proposed double-resampling method.

\begin{table}[h]
\begin{center}
    \caption{\label{tab:HR} Estimated $\beta_{0}$ and hazard ratio of comparing
"carboplatin + paclitaxel" to "erlotinib" }
{ \begin{tabular}{llccc} 
\hline
      &            & $\widehat{\beta}$  & Estimated Hazard Ratio & 95\% Confidence Interval  \\ 
\hline
Naive &            & -0.085             & 0.918                  & (0.851, 0.991)            \\
IPW   &            & -0.065             & 0.937                  & (0.850, 1.033)            \\
AIPW  &            & -0.075             & 0.928                  & (0.842, 1.023)              \\
X.M   &            & -0.076             & 0.927                  & (0.836, 1.028)             \\
PSM  & asymp      & -0.058             & 0.944                  & (0.839, 1.062)            \\
(M=1) & double-rsp & -0.058             & 0.944                  & (0.830, 1.074)             \\
PSM  & asymp      & -0.009             & 0.991                  & (0.868, 1.131)             \\
(M=5) & double-rsp & -0.009             & 0.991                  & (0.903, 1.087)             \\
\hline
\end{tabular}}
\\
\end{center}
 Note: The proposed PSM estimator is calculated under the number of matches $M=1$ and 5. For each PSM estimators, its variance is estimated by the proposed asymptotic variance, "asymp" and the proposed double-resampling method,  "double-rsp". 
\end{table}


All adjusted methods give larger hazard ratio estimates than the unadjusted
naive method. Although the point estimates of the PSM approach with $M=1$ and $M=5$ indicate that
"carboplatin + paclitaxel" might be slightly more advantageous, the 95\% confidence intervals using both inference approaches include
1, implying insufficient evidence for a statistically
significant difference at the 0.05 level in the effectiveness comparison. Similar to the simulation results, the PSM estimator with $M=5$ provides narrower interval than the one with $M=1$ for the double-resampling approach. The 95\% confidence intervals of the other adjusted approaches presented in Table \ref{tab:HR}, including IPW, AIPW, and matching on covariates, also contain 1, reaching the same conclusion as the PSM estimators. 


\section{Discussion and future studies}\label{sec:Discussions}

PSM is prevailing in practice to handle confounding
in observational studies. We establish the statistical properties
of the PSM estimator of the marginal causal
hazard ratio based on matching with replacement and a fixed number
of matches. We also propose a double-resampling technique for variance
estimation that takes into account the uncertainty due to propensity
score estimation prior to matching. Existing simulation studies have indicated that for estimation of the average treatment effect, IPW or AIPW can perform unstably when extreme values of the estimated propensity scores are present, and instead PSM or alternative approaches such as the overlap weights should be considered and leveraged. Our simulation results echo the previous findings in the survival context for estimation of the marginal hazard ratio.

The theoretical results established in this article hold for any fixed number of matches $M$. We illustrated two different choices for $M$ in the simulation study and real data analysis in order to show the validity of the resampling-based variance estimator. Intuitively, different choices of $M$ affect the PSM estimator through a bias-variance trade-off. In practice, often a small $M$ is recommended, since the increase in bias is often more significant than the reduction in the variance \citep{abadie2006large,austin2010statistical}. Following the cross-validation idea for the matching estimator with continuous outcomes proposed in \citet{xu2023advances}, we leave the development of an optimal way of choosing $M$ for the PSM estimator as a topic for future research.

    In this work, we derived the asymptotic distribution of the PSM estimator of the marginal hazard ratio assuming a generalized linear model for the propensity score. Notably, the same proof technique can potentially be extended to cases when the propensity scores are estimated assuming certain semiparametric or nonparametric models (e.g., single-
or multiple-index models; \citet{huang2017joint}). Take the semiparametric single-index model for example. In this case, the key insight
is that PSM does not rely on the exact functional form
of the propensity score model but a sufficient dimension reduction
of the mean space of $A_{i}$ given $X_{i}$. We have $e(X_{i})=e(X_{i}^{T}\theta_{0})$,
where $\theta_{0}\in \mathbb{R}^{p}$ is a vector of unknown parameters
and $e(\cdot)$ is left unspecified and does not require a restrictive
parametric model assumption. Although the link function does not permit a root-$n$ consistent estimator, the index $X^{\T}\theta$ enables a root-$n$ consistent
estimator, denoted as $X^{\T}\widehat{\theta}$ \citep{huang2017joint}.
Then it suffices to implement the PSM based on $X^{\T}\widehat{\theta}$. More specifically, in the proof of Theorem \ref{Thm2}, under a parametric propensity score
model, $P^{\theta_{0}}$ for (\ref{eq:LAN}) is
naturally the probability measure governed by the likelihood function
of $\theta_{0}$. Under the single- or multiple index propensity
score model, following \citet{andreou2012alternative}, we can formulate
a working model for $P^{\theta_{0}}$ based on the asymptotic distribution of $X^{\T}\widehat{\theta}$.
The remaining steps for the proof would remain the same and the inferential framework can be carried over. When assuming other nonparametric machine learning models for the propensity score, the same procedure for implementing the PSM estimator can still be applied, and we leave the development of inferential frameworks in those cases as future endeavors.

It is important to draw connections between clinical trials and observational
studies from both design and analysis perspectives, which highlights
the advantages of PSM and also motivates several
future research directions. Similar to clinical trial designs, the
matching step uses only the covariate and treatment information and
does not touch the outcome data. Therefore, it mitigates the possibility
of data snooping and dredging. As discussed in the introduction, PSM alone emulates a completely randomized trial, which
may not be the most efficient. Stratified block randomization is often
used to improve the complete randomization in clinical trials.
Thus, we can imagine that combining stratification and PSM can also improve PSM alone in
observational studies. Moreover, in trial data analysis, instead of
a simple analysis of the outcome data, ANCOVA can be used to borrow
information from auxiliary information. Thus, it is important to continue
the development of more efficient analysis methods, such as general M and Z estimators \citep{van1998asymptotic}, for the matched
observational data. We leave these topics to future research.


This work focuses on estimating the causal PH ratio, a scalar estimand,
that summarizes the treatment effect over a certain period of time.
There are many other treatment effect estimands for survival outcomes,
such as the restrictive mean survival times, restrictive mean lost
times, the difference in survival medians, and so on \citep{yang2020smim}.
In the context of survival analyses, \citet{chen2001causal} proposed
a regression model approach to estimate the average causal effect
of restricted mean survival times. \cite{xie2005adjusted} developed
the adjusted Kaplan-Merier estimators of treatment-specific survival
functions using IPW. \cite{zhang2012double}
combined the regression method of \cite{chen2001causal} and the inverse
probability weighted Nelson-Aalen estimator, resulting in a doubly
robust estimator of the average causal effect of restricted mean survival
times. \cite{diaz2019statistical} proposed data-adaptive doubly
robust estimators of treatment-specific survival functions. In our
future work, we will develop matching estimators of general causal
estimands for survival outcomes and compare them with existing
approaches. %

\section*{Acknowledgments}
Yang is supported in part by U.S. National Science Foundation and National Institute of
Health.

\bibliographystyle{apalike}
\bibliography{refs}

\begin{thebibliography}{}

\bibitem[Abadie and Imbens, 2006]{abadie2006large}
Abadie, A. and Imbens, G.~W. (2006).
\newblock Large sample properties of matching estimators for average treatment
  effects.
\newblock {\em Econometrica}, 74:235--267.

\bibitem[Abadie and Imbens, 2008]{abadie2008failure}
Abadie, A. and Imbens, G.~W. (2008).
\newblock On the failure of the bootstrap for matching estimators.
\newblock {\em Econometrica}, 76:1537--1557.

\bibitem[Abadie and Imbens, 2016]{abadie2016matching}
Abadie, A. and Imbens, G.~W. (2016).
\newblock Matching on the estimated propensity score.
\newblock {\em Econometrica}, 84:781--807.

\bibitem[Adusumilli, 2022]{adusumilli2022bootstrap}
Adusumilli, K. (2022).
\newblock Bootstrap inference for propensity score matching.
\newblock {\em Working Paper}.

\bibitem[Andreou and Werker, 2012]{andreou2012alternative}
Andreou, E. and Werker, B.~J. (2012).
\newblock An alternative asymptotic analysis of residual-based statistics.
\newblock {\em Rev Econ Stat}, 94:88--99.

\bibitem[Austin, 2010]{austin2010statistical}
Austin, P.~C. (2010).
\newblock Statistical criteria for selecting the optimal number of untreated
  subjects matched to each treated subject when using many-to-one matching on
  the propensity score.
\newblock {\em Am J Epidemiol}, 172(9):1092--1097.

\bibitem[Austin, 2013]{austin2013performance}
Austin, P.~C. (2013).
\newblock The performance of different propensity score methods for estimating
  marginal hazard ratios.
\newblock {\em Stat Med}, 32:2837--2849.

\bibitem[Austin and Cafri, 2020]{austin2020variance}
Austin, P.~C. and Cafri, G. (2020).
\newblock Variance estimation when using propensity-score matching with
  replacement with survival or time-to-event outcomes.
\newblock {\em Stat Med}, 39(11):1623--1640.

\bibitem[Austin et~al., 2007]{austin2007conditioning}
Austin, P.~C., Grootendorst, P., Normand, S.-L.~T., and Anderson, G.~M. (2007).
\newblock Conditioning on the propensity score can result in biased estimation
  of common measures of treatment effect: a monte carlo study.
\newblock {\em Stat Med}, 26(4):754--768.

\bibitem[Austin and Small, 2014]{austin2014use}
Austin, P.~C. and Small, D.~S. (2014).
\newblock The use of bootstrapping when using propensity-score matching without
  replacement: a simulation study.
\newblock {\em Stat Med}, 33:4306--4319.

\bibitem[Austin and Stuart, 2017]{austin2017performance}
Austin, P.~C. and Stuart, E.~A. (2017).
\newblock The performance of inverse probability of treatment weighting and
  full matching on the propensity score in the presence of model
  misspecification when estimating the effect of treatment on survival
  outcomes.
\newblock {\em Stat Methods Med Res}, 26:1654--1670.

\bibitem[Bang and Robins, 2005]{bang2005doubly}
Bang, H. and Robins, J.~M. (2005).
\newblock Doubly robust estimation in missing data and causal inference models.
\newblock {\em Biometrics}, 61:962--973.

\bibitem[Busso et~al., 2014]{busso2014new}
Busso, M., DiNardo, J., and McCrary, J. (2014).
\newblock New evidence on the finite sample properties of propensity score
  reweighting and matching estimators.
\newblock {\em Rev Econ Stat}, 96(5):885--897.

\bibitem[Chen and Tsiatis, 2001]{chen2001causal}
Chen, P.-Y. and Tsiatis, A.~A. (2001).
\newblock Causal inference on the difference of the restricted mean lifetime
  between two groups.
\newblock {\em Biometrics}, 57:1030--1038.

\bibitem[Cox, 1972]{cox1972regression}
Cox, D.~R. (1972).
\newblock Regression models and life-tables.
\newblock {\em J R Stat Soc Series B Stat Methodol}, 34:187--220.

\bibitem[Crump et~al., 2009]{crump2009dealing}
Crump, R.~K., Hotz, V.~J., Imbens, G.~W., and Mitnik, O.~A. (2009).
\newblock Dealing with limited overlap in estimation of average treatment
  effects.
\newblock {\em Biometrika}, 96:187--199.

\bibitem[Cui et~al., 2018]{cui2018application}
Cui, Z.~L., Hess, L.~M., Goodloe, R., and Faries, D. (2018).
\newblock Application and comparison of generalized propensity score matching
  versus pairwise propensity score matching.
\newblock {\em J Comp Eff Res}, 7:923--934.

\bibitem[de~Rooij, 2018]{de2018transitional}
de~Rooij, M. (2018).
\newblock Transitional modeling of experimental longitudinal data with missing
  values.
\newblock {\em Adv Data Anal Classif}, 12:107--130.

\bibitem[Dehejia and Wahba, 2002]{dehejia2002propensity}
Dehejia, R.~H. and Wahba, S. (2002).
\newblock Propensity score-matching methods for nonexperimental causal studies.
\newblock {\em Rev Econ Stat}, 84(1):151--161.

\bibitem[D{\'\i}az, 2019]{diaz2019statistical}
D{\'\i}az, I. (2019).
\newblock Statistical inference for data-adaptive doubly robust estimators with
  survival outcomes.
\newblock {\em Stat Med}, 38:2735--2748.

\bibitem[Fr{\"o}lich, 2004]{frolich2004finite}
Fr{\"o}lich, M. (2004).
\newblock Finite-sample properties of propensity-score matching and weighting
  estimators.
\newblock {\em Rev Econ Stat}, 86:77--90.

\bibitem[Gayat et~al., 2012]{gayat2012propensity}
Gayat, E., Resche-Rigon, M., Mary, J.-Y., and Porcher, R. (2012).
\newblock Propensity score applied to survival data analysis through
  proportional hazards models: a monte carlo study.
\newblock {\em Pharm Stat}, 11:222--229.

\bibitem[Greifer and Stuart, 2021]{greifer2021matching}
Greifer, N. and Stuart, E.~A. (2021).
\newblock Matching methods for confounder adjustment: an addition to the
  epidemiologist's toolbox.
\newblock {\em Epidemiol Rev}, 43(1):118--129.

\bibitem[Hern{\'a}n et~al., 2000]{hernan2000marginal}
Hern{\'a}n, M.~{\'A}., Brumback, B., and Robins, J.~M. (2000).
\newblock Marginal structural models to estimate the causal effect of
  zidovudine on the survival of {HIV}-positive men.
\newblock {\em Epidemiology}, 11:561--570.

\bibitem[Huang and Chan, 2017]{huang2017joint}
Huang, M.-Y. and Chan, K. C.~G. (2017).
\newblock Joint sufficient dimension reduction and estimation of conditional
  and average treatment effects.
\newblock {\em Biometrika}, 104:583--596.

\bibitem[Imbens and Rubin, 2015]{imbens2015causal}
Imbens, G.~W. and Rubin, D.~B. (2015).
\newblock {\em {Causal Inference in Statistics, Social, and Biomedical
  Sciences}}.
\newblock Cambridge University Press, Cambridge UK.

\bibitem[Kang and Schafer, 2007]{kang2007demystifying}
Kang, J.~D. and Schafer, J.~L. (2007).
\newblock Demystifying double robustness: A comparison of alternative
  strategies for estimating a population mean from incomplete data.
\newblock {\em Stat Sci}, 22:523--539.

\bibitem[Le~Cam and Yang, 1990]{le1990asymptotics}
Le~Cam, L. and Yang, G.~L. (1990).
\newblock {\em Asymptotics in Statistics: Some Basic Concepts}.
\newblock Springer: Berlin.

\bibitem[Li et~al., 2017]{li2016balancing}
Li, F., Morgan, K.~L., and Zaslavsky, A.~M. (2017).
\newblock Balancing covariates via propensity score weighting.
\newblock {\em J Am Stat Assoc}, 113:390--400.

\bibitem[Martinussen and Vansteelandt, 2013]{martinussen2013collapsibility}
Martinussen, T. and Vansteelandt, S. (2013).
\newblock On collapsibility and confounding bias in cox and aalen regression
  models.
\newblock {\em Lifetime Data Anal}, 19(3):279--296.

\bibitem[Otsu and Rai, 2017]{otsu2017bootstrap}
Otsu, T. and Rai, Y. (2017).
\newblock Bootstrap inference of matching estimators for average treatment
  effects.
\newblock {\em J Am Stat Assoc}, 112(520):1720--1732.

\bibitem[Robins, 2004]{robins2004optimal}
Robins, J.~M. (2004).
\newblock Optimal structural nested models for optimal sequential decisions.
\newblock In {\em Proceedings of the second seattle Symposium in
  Biostatistics}, pages 189--326, New York. Springer.

\bibitem[Rubin, 2006]{rubin2006matched}
Rubin, D.~B. (2006).
\newblock {\em {Matched Sampling for Causal Effects}}.
\newblock Cambridge University Press, Cambridge, England.

\bibitem[Stuart, 2010]{stuart2010matching}
Stuart, E.~A. (2010).
\newblock Matching methods for causal inference: A review and a look forward.
\newblock {\em Stat Sci}, 25:1--21.

\bibitem[Tchetgen~Tchetgen and Robins, 2012]{tchetgen2012parametrization}
Tchetgen~Tchetgen, E.~J. and Robins, J. (2012).
\newblock On parametrization, robustness and sensitivity analysis in a marginal
  structural cox proportional hazards model for point exposure.
\newblock {\em Stat Probab Lett}, 82:907--915.

\bibitem[Vaart, 1998]{van1998asymptotic}
Vaart, A. W. v.~d. (1998).
\newblock {\em Asymptotic Statistics}.
\newblock Cambridge University Press.

\bibitem[Vansteelandt and Daniel, 2014]{vansteelandt2014regression}
Vansteelandt, S. and Daniel, R.~M. (2014).
\newblock On regression adjustment for the propensity score.
\newblock {\em Stat Med}, 33(23):4053--4072.

\bibitem[Waernbaum, 2012]{waernbaum2012model}
Waernbaum, I. (2012).
\newblock Model misspecification and robustness in causal inference: comparing
  matching with doubly robust estimation.
\newblock {\em Stat Med}, 31(15):1572--1581.

\bibitem[Wang et~al., 2022]{wang2022instrumental}
Wang, L., Tchetgen~Tchetgen, E., Martinussen, T., and Vansteelandt, S. (2022).
\newblock Instrumental variable estimation of the causal hazard ratio.
\newblock {\em Biometrics}.

\bibitem[Williamson et~al., 2012]{williamson2012propensity}
Williamson, E., Morley, R., Lucas, A., and Carpenter, J. (2012).
\newblock Propensity scores: from naive enthusiasm to intuitive understanding.
\newblock {\em Stat Methods Med Res}, 21(3):273--293.

\bibitem[Xie and Liu, 2005]{xie2005adjusted}
Xie, J. and Liu, C. (2005).
\newblock Adjusted {Kaplan--Meier} estimator and log-rank test with inverse
  probability of treatment weighting for survival data.
\newblock {\em Stat Med}, 24:3089--3110.

\bibitem[Xu, 2023]{xu2023advances}
Xu, T. (2023).
\newblock {\em Advances in Causal Inference and the Study of Interlocus Gene
  Conversion}.
\newblock PhD thesis, North Carolina State University, Raleigh, NC.

\bibitem[Yang and Ding, 2018]{yang2018asymptotic}
Yang, S. and Ding, P. (2018).
\newblock Asymptotic inference of causal effects with observational studies
  trimmed by the estimated propensity scores.
\newblock {\em Biometrika}, 105:487--493.

\bibitem[Yang et~al., 2016]{yang2016propensity}
Yang, S., Imbens, G.~W., Cui, Z., Faries, D.~E., and Kadziola, Z. (2016).
\newblock Propensity score matching and subclassification in observational
  studies with multi-level treatments.
\newblock {\em Biometrics}, 72:1055--1065.

\bibitem[Yang et~al., 2020]{yang2020smim}
Yang, S., Zhang, Y., Liu, G.~F., and Guan, Q. (2020).
\newblock Smim: a unified framework of survival sensitivity analysis using
  multiple imputation and martingale.
\newblock {\em arXiv preprint arXiv:2007.02339}.

\bibitem[Zhang and Schaubel, 2012]{zhang2012double}
Zhang, M. and Schaubel, D.~E. (2012).
\newblock Double-robust semiparametric estimator for differences in restricted
  mean lifetimes in observational studies.
\newblock {\em Biometrics}, 68:999--1009.

\bibitem[Zhao, 2023]{zhao2023advances}
Zhao, H. (2023).
\newblock {\em Advances in Matching Methods for Causal Inference with Multiple
  Treatments}.
\newblock PhD thesis, North Carolina State University, Raleigh, NC.

\bibitem[Zhao et~al., 2022]{zhao2022double}
Zhao, H., Zhang, X., and Yang, S. (2022).
\newblock Double score matching in observational studies with multi-level
  treatments.
\newblock {\em Commun Stat Simul Comput}, pages 1--17.

\end{thebibliography}

\pagebreak{}
\appendix

\section*{APPENDIX}

\section{Comparison between AIPW and PSM for estimating the marginal hazard ratio}

The AIPW estimator is also a popular approach for estimating the average treatment effect for a continuous outcome and has been adapted to estimate the marginal hazard ratio in the survival context \cite{tchetgen2012parametrization}. The AIPW estimator is consistent as long as either the propensity score model or the outcome model is correctly specified. In comparison, PSM does not rely on the outcome model and is consistent if the propensity score model is correctly specified. The AIPW estimator is semiparametrically efficient when both models are correctly specified. In comparison, based on Theorem 1, the PSM estimator based on a correctly specified propensity score model will be less efficient than AIPW asymptotically as it does not in general attain the semiparametric efficiency bound. However, as shown in our simulation study, in cases where the overlap of covariate distributions between the two treatment groups is poor, PSM can yield more stable marginal hazard ratio estimates than AIPW in finite samples. This phenomenon is in line with findings from earlier numerical studies involving survival or continuous outcomes \citep{busso2014new,austin2017performance}. This is in part because AIPW weights the outcomes by the inverse of the estimated propensity scores, whereas PSM reuses the original value of the outcomes regardless of how extreme the propensity scores are. \\ \indent Trimming offers a general solution to violation of the overlap assumption by excluding units with extreme propensity scores from the analysis \citep{crump2009dealing,yang2018asymptotic}. However, trimming methods necessary alter the target estimand and restrict inference to regions of the covariate space with sufficient overlap. PSM is more resistant to the violation of the overlap assumption because even for units with extreme propensity scores, their missing potential outcomes can still be imputed with minimal bias as long as they have close matches from the opposite treatment group. Therefore, trimming is not needed for PSM except maybe when the matching discrepancies become noticeably large.

\section{Asymptotic variance estimation}

In this section, we discuss estimation of the large sample variances of $\widehat{\beta}$ adjusting for first step estimation of the propensity score. Recall that we have that 
$$V_2=\{A(\beta_0)\}^{-1}\{V_{G}-c^{\T}\mathcal{I}_{\theta_{0}}^{-1}c\}\{A(\beta_0)\}^{-1}.$$
We will estimate each component on the right hand side of the equation separately. We first estimate the Fisher information $\mathcal{I}_{\theta_{0}}$ using
$$\widehat{\mathcal{I}}_{\theta_{0}}=\frac{1}{n} \sum_{i=1}^n  \frac{\dot{e}^2(X_i^{\T}\widehat{\theta})}{e(X_i^{\T}\widehat{\theta})(1-e(X_i^{\T}\widehat{\theta}))}   X_i X_i^{\T}.$$

Let $m_k\{\omega, e(X_i ^{\T}\widehat{\theta})\}$ denote the index of the $k$-th nearest neighbor matched to unit $i$ based on the estimated propensity scores. For estimation of $V_{G}$, we first create an imputed dataset $\left\{H_{i 1}^*(\omega), H_{i 2}^*(\omega)\right\}_{i=1}^n$ for $\omega=0,1$, where
\[H_{i 1}^*(\omega)=\begin{cases}
H_i(\omega) & \text { if } W_i=\omega \\
H_{m_1\left\{\omega, e(X_i ^{\T}\widehat{\theta})\right\}}(\omega) & \text { if } W_i \neq \omega
\end{cases}\]
and
\[H_{i 2}^*(\omega)=\begin{cases}
H_{m_1\left\{\omega, e(X_i ^{\T}\widehat{\theta})\right\}}(\omega) & \text { if } W_i=\omega \\
H_{m_2\left\{\omega, e(X_i ^{\T}\widehat{\theta})\right\}}(\omega) & \text { if } W_i \neq \omega.
\end{cases}\]
Then, $\widetilde{V}_{G}$ can be estimated by
$$\widehat{V}_{G}=\frac{1}{n} \sum_{i=1}^n\left\{\sum_{\omega=0}^1 H_{i 1}^*(\omega)\right\}^{ 2}+\frac{1}{n} \sum_{i=1}^n\left\{k_{\widehat{\theta},i}\left(W_i\right)+k_{\widehat{\theta},i}\left(W_i\right)^2\right\} \widehat{\sigma}_i^2,$$
where $\widehat{\sigma}_i^2=\sum_{k=1}^2\left[H_{i k}^*(\omega)-\frac{1}{2}\left\{H_{i 1}^*(\omega)+H_{i 2}^*(\omega)\right\}\right]^{ 2}=\frac{1}{2}\left\{H_{i 1}^*(\omega)-H_{i 2}^*(\omega)\right\}^{ 2}.$

Next we can construct an estimator of $c$ by averaging over the sample: 
$$\widehat{c}=\frac{1}{n}\sum_{i=1}^n \left[\frac{\widehat{\cov}\left\{ X_i,\mu_{H}(1,X_i)\mid e(X_i^{\T}\theta_{0})\right\} }{e(X_i^{\T}\widehat{\theta})}+\frac{\widehat{\cov}\left\{ X_i,\mu_{H}(0,X_i)\mid e(X_i^{\T}\theta_{0})\right\} }{1-e(X_i^{\T}\widehat{\theta})}\right]\dot{e}(X_i^{\T}\widehat{\theta}) .$$
\sloppy For estimation of the conditional covariance, we follow the same matching procedure to create an imputed dataset $\left\{X_{i 1}^*(\omega), X_{i 2}^*(\omega)\right\}_{i=1}^n$. Then $\widehat{\cov}\left\{X_i, \mu_{H}\left(\omega, X_i\right)  \mid e_\omega(X_i^{\T}\theta_{0})\right\}$ can be estimated by
$\frac{1}{2}\left\{X_{i 1}^*(\omega)-X_{i 2}^*(\omega)\right\}\left\{H_{i 1}^*(\omega)-H_{i 2}^*(\omega)\right\}.$

Finally, to estimate $A(\beta_0)$, we use
$$\widehat{A}(\beta_0)=\frac{1}{n} \sum_{i=1}^n \sum_{k=1}^K\left[\left\{\widehat{Q}(\widehat{\beta}, t_k)-\widehat{Q}(\widehat{\beta}, t_k)^{ 2}\right\}\left\{ \sum_{\omega=0}^1 \de \overline{N}_i^{*(\omega)}\left(t_k\right)\right\}\right],$$
where $\{t_1,...t_K\}$ are distinct observed time points. Putting everything together, our final estimator of the asymptotic variance is 
$$
\widehat{V}_{2}=\{\widehat{A}(\beta_0)\}^{-1}\{\widehat{V}_{G}-\widehat{c}^{\T}\widehat{\mathcal{I}}_{\theta_{0}}^{-1}\widehat{c}\}\{\widehat{A}(\beta_0)\}^{-1}.
$$

\section{Regularity conditions and Lemmas}

In this section, we provide the regularity conditions and lemmas.
For simplicity, we introduce more notations. 
Let $\mathcal{N}\equiv \{\theta: ||\theta-\theta_0||<\epsilon\}$ be a neighborhood of $\theta_0$  given an $\epsilon>0$. We use
a generalized linear specification for the propensity score, $e(x) = e(X^{\T}\theta)$ where $e(\cdot)$ is a link function. Moreover, denote $S_{n}(\beta,t)=n^{-1}\sum_{j=1}^{n}\{1+k_{\widehat{\theta},i}/M\}I(W_{j}=\omega)\exp(\beta W_{i})Y_{j}(t)$ and $S_{n}'(\beta,t)=\partial S_n(\beta,t)/\partial \beta =n^{-1}\sum_{j=1}^{n}\{1+k_{\widehat{\theta},i}/M\}I(W_{j}=\omega)\exp(\beta W_{i})Y_{j}(t)W_{i}$, then $\widehat{Q}(\beta,t)=S_{n}(\beta,t)/S_{n}'(\beta,t)$. 

\begin{assumption}\label{asump:regua} 
The following regularity conditions hold:
\begin{enumerate}[label=(\roman*)]
    \item  $\theta_0 \in int(\Theta)$ with $\Theta$ compact, $X$ has a bounded support and $\E[X^{\T}X]$ is non-singular; 
    \item  $e(\cdot):\mathbb{R}\rightarrow (0,1)$ is twice continuously differentiable with strictly bounded first and second derivatives, $\dot{e}(\cdot)$ and $\ddot{e}(\cdot)$ where $\dot{e}(\cdot)$ is strictly positive; 
\item  $\forall \theta \in \mathcal{N}$, the random variable $e(X^{\T}\theta)$ is continuously distributed with interval support, and its pdf $g_\theta(\cdot)$ is uniformly Lipschitz continuous over $\mathcal{\theta}$; 
\item  there exists a component of $X$ that is continuously distributed, has nonzero coefficient in $\theta_0$, and has a continuous density function conditional on the rest of $X$; 
\item $\forall \theta \in \mathcal{N}$ and $\omega=0,1$, $\mu_{H}(\omega,X)$ and $\sigma^2_{H}(\omega,X)$ is Lipschitz-continuous in $p$ with the Lipschitz constants independent of $\theta$; 
\item $E\left(|H_{i}(\omega)|^{4+\delta}|W_{i}=\omega,X_{i}=x\right)$
is uniformly bounded over the support of $X$, for some $\delta>0$; 
\item $\E\left\{ \de M^{(1)}(t)\mid p\right\} $ is Lipschiz continuous in $p$; 
\item There exists $\tau>0$ is such that $\int_0^\tau \lambda_0(t)d(t)<\infty$; 
\item there exists a neighborhood $\mathcal{B}$ of $\beta_0$ and functions $s_0(\beta,t)$ and $s_1(\beta,t)$such that $sup_{t\in [0,\tau],\beta\in \mathcal{B}}||S_n(\beta,t)-s_0(\beta,t)||\rightarrow_p 0$ and $sup_{t\in [0,\tau],\beta\in \mathcal{B}}||S_n'(\beta,t)-s_1(\beta,t)||\rightarrow_p 0$ as $n\rightarrow \infty$ ; 
\item the function $s_0(\beta,t)$ and $s_1(\beta,t)$ is bounded away from 0 on $\mathcal{B}\times [0,\tau]$; 
\item $\forall \epsilon>0$, there exists $\delta>0$ such that $||\beta-\beta_0||< \delta$ implies $|s_l(\beta,t)-s_l(\beta_0,t)|<\epsilon$, $\forall t \in [0,\tau]$, $l=0,1$;
\item $n^{-1}\de G_{n}(\tilde{\beta})/\de\beta>0$ for $\forall \tilde{\beta} \in \mathcal{B}$. 
\end{enumerate}
\end{assumption}

The regularity conditions (i)-(vi) are adopted from \cite{abadie2016matching} and \cite{adusumilli2022bootstrap},  modified for survival outcomes. The regularity conditions (vii)-(xi) are standard in the survival analysis literature and are often
assumed for technical convenience. We note here that since $\beta$ is a scale and $\omega$ only takes values in $\{0,1\}$, the first and second derivatives of $S_n(\beta,t)$ with respect to $\beta$ are the same. Therefore, it suffices to impose the regulation conditions on the first derivative in (vii)-(xi).

\begin{assumption}\label{asump:regua2}
Additional regularity conditions for variance estimation hold:
\begin{enumerate}[label=(\roman*)]
\item  Uniformly over all $\theta \in \mathcal{N}$, it holds under $\mathbb{P}_0$, ${N}^{-1}\sum_{i=1}^N\left(\widehat{e}_{1i}(\theta)-e_{1i}(\theta)\right)^2=o_p(N^{-\xi})$ and ${N}^{-1}\sum_{i=1}^N\left(\widehat{e}_{2i}(W_i;\theta)-e_{2i}(W_i;\theta)\right)^2=o_p(N^{-\xi})$ for some $\xi>0$; 
\item the number of quantile partitions satisfies $q_N\rightarrow \infty$ and $q_N^{2+\eta}\rightarrow 0$ as $N\rightarrow \infty$ for some $\eta>0$.   
\end{enumerate}
\end{assumption}
The Assumption \ref{asump:regua2} are adopted from \cite{adusumilli2022bootstrap} to ensure the validity of the double-resampling bootstrap estimator. 

Lemma \ref{lemma A.11} below is Lemma S.11 in \cite{abadie2016matching}, which is useful in the proofs of our results.
\begin{lemma}\label{lemma A.11}Suppose that $(W_{1},X_{1})$,...,$(W_{n},X_{n})$
are independent and identically distributed, where $X_i$ is a scalar
continuous variable with a bounded support. Suppose also that $\sigma_{H}^{2}(\omega,x)$
is uniformly bounded over the support for $X$. Let $n_{\omega}=\sum_{i=1}^{n}I(W_{i}=\omega)$ be the number of individuals
received treatment $\omega$, and $p^*=\pr(W=1)>0$ and $f_{\omega}(X)$ be the density function of the scalar
continuous variable X when $W=\omega$.Then, for a given
$\omega$, 
\[
\frac{1}{n_{\omega}}\sum_{i=1}^{n}I(W_{i}=\omega)\sigma_{H}^{2}(\omega,X_{i})k_{i}\rightarrow M \E\left\{ \sigma_{H}^{2}(\omega,X)\left(\frac{p^*}{1-p^*}\right)^{1-2\omega}\frac{f_{1-\omega}(X)}{f_\omega(X)}\mid W_{i}=\omega\right\} ,
\]
and 
\begin{eqnarray*}
\frac{1}{n_{\omega}}\sum_{i=1}^{n}I(W_{i}=\omega)\sigma_{H}^{2}(\omega,X_i)k_{i}^{2} & \rightarrow  M\E\left\{ \sigma_{H}^{2}(\omega,X)\left(\frac{p^*}{1-p^*}\right)^{1-2\omega}\frac{f_{1-\omega}(X)}{f_\omega(X)}\mid W_{i}=\omega\right\} \\
 &  +\frac{M(2M+1)}{2}\E\left\{\sigma_{H}^{2}(\omega,X)\left[\left(\frac{p^*}{1-p^*}\right)^{1-2\omega}\frac{f_{1-\omega}(X)}{f_\omega(X)}\right]^2\mid W_{i}=\omega\right\},
\end{eqnarray*}
in probability, as $n\rightarrow\infty$.

\end{lemma}

\section{Proof of the asymptotic unbiasedness of Equation (\ref{eq:linear S})}

This section includes three parts that follow the similar logic of proof. 
The first and the second parts provide some results useful for later sections.
The proof for the asymptotic unbiasedness of $n^{-1}G_{n}(\beta_{0})$ is located in the third part.

For $\omega=0,1,$ define $\de M^{(\omega)}(t)=\de N^{(\omega)}(t)-\de\Lambda_{0}(t)\exp(\beta_{0}\omega)Y^{(\omega)}(t)$.
From the standard theory for the counting process, $\de M^{(\omega)}(t)$
is a martingale process with respect to the population and its baseline hazard is $\Lambda_0(t)$. Next we will prove that 
$$
I(W_i=\omega)\{1+k_i/M\}\left\{\de N_i(t)-\de\Lambda_{0}(t)\exp(\beta_0\omega)Y_{i}(t)\right\}
$$
 is a martingale for the imputed pseudo-population which means that the imputed pseudo-population has similar covariates distribution with the target population.
First, we show that for $\omega=0,1,$
\begin{align}
 &  n^{-1}\sum_{i=1}^{n}I(W_{i}=\omega)\{1+k_{i}/M\}\left\{ \de N_{i}(t)-\de\Lambda_{0}(t)\exp(\beta_{0}\omega)Y_{i}(t)\right\} \nonumber \\
 & \rightarrow  \E\left\{ \de N^{(\omega)}(t)-\de\Lambda_{0}(t)\exp(\beta_{0}\omega)Y^{(\omega)}(t)\right\} =\E\left\{ \de M^{(\omega)}(t)\right\} ,\label{eq:lemma1}
\end{align}
as $n\rightarrow\infty$. We show (\ref{eq:lemma1}) for $\omega=1$.
The proof for $\omega=0$ is similar and therefore omitted. We express
(\ref{eq:lemma1}) for $\omega=1$ as
\begin{align*}
 &   n^{-1}\sum_{i=1}^{n}W_{i}\{1+k_{i}/M\}\left\{ \de N_{i}(t)-\de\Lambda_{0}(t)\exp(\beta_{0})Y_{i}(t)\right\} -\E\left\{ \de M^{(1)}(t)\right\} \\
 & =  n^{-1}\sum_{i=1}^{n}W_{i}\{1+k_{i}/M\}\left\{ \de N_{i}^{(1)}(t)-\de\Lambda_{0}(t)\exp(\beta_{0})Y_{i}^{(1)}(t)\right\} -\E\left\{ \de M^{(1)}(t)\right\} \\
 & =  n^{-1}\sum_{i=1}^{n}W_{i}\{1+k_{i}/M\}\de M_{i}^{(1)}(t)-\E\left\{ \de M^{(1)}(t)\right\} \\
 & =  n^{-1}\sum_{i=1}^{n}W_{i}\{1+k_{i}/M\}\left[\de M_{i}^{(1)}(t)-\E\left\{ \de M^{(1)}(t)\mid e(X_{i})\right\} \right]\\
 & \;\quad  +n^{-1}\sum_{i=1}^{n}(1-W_{i})\left[\E\left\{ \de M^{(1)}(t)\mid e(X_{m\{1,e(X_{i})\}})\right\} -\E\left\{ \de M^{(1)}(t)\mid e(X_{i})\right\} \right]\\
 & \; \quad  +n^{-1}\sum_{i=1}^{n}\E\left\{ \de M^{(1)}(t)\mid e(X_{i})\right\} -\E\left\{ \de M^{(1)}(t)\right\} \\
 & =  T_{1}+T_{2}+T_{3},
\end{align*}
where the second line follows by the consistent assumption, and
\begin{eqnarray}
T_{1} & = & n^{-1}\sum_{i=1}^{n}W_{i}\{1+k_{i}/M\}\left[\de M_{i}^{(1)}(t)-\E\left\{ \de M^{(1)}(t)\mid e(X_{i})\right\} \right],\nonumber \\
T_{2} & = & n^{-1}\sum_{i=1}^{n}(1-W_{i})\left[\E\left\{ \de M^{(1)}(t)\mid e(X_{m\{1,e(X_{i})\}})\right\} -\E\left\{ \de M^{(1)}(t)\mid e(X_{i})\right\} \right],\label{eq:T2}\\
T_{3} & = & n^{-1}\sum_{i=1}^{n}\E\left\{ \de M^{(1)}(t)\mid e(X_{i})\right\} -\E\left\{ \de M^{(1)}(t)\right\} .\nonumber 
\end{eqnarray}
Under Assumption \ref{asump:regua} (i) - (vi), \cite{abadie2006large} showed
that $k_{i}^{\delta}$ is bounded almost surely for any $\delta>0$,
and the discrepancy due to matching is $|{M}^{-1}\sum_{j\in \mathcal{J}_M(1,e(X_i))}e(X_j)-e(X_{i})|=O_p(n^{-1})$
for a scalar $e(X)$. It follows that $T_{1}$ is consistent for zero.
Moreover, under Assumption \ref{asump:regua} (vi), $T_{2}$ is consistent for zero. Lastly,
by the strong law of large numbers, $T_{3}$ is consistent for zero.
Therefore, (\ref{eq:lemma1}) follows. 
Since $n^{-1}\sumin I(W_i=\omega)(1+k_i/M) \left\{\omega-\widehat{Q}(\beta_0,t) \right\}\de M_i^{(\omega)}(t)$ is bounded, by dominated convergence theorem,
\begin{align}
n^{-1}G_n(\beta_{0})&=n^{-1}\sum_{\omega=0}^1\sumin \intin I(W_i=\omega) \{1+k_i/M\}\left\{\omega-\widehat{Q}(\beta_0,t) \right\}\de M_i^{(\omega)}(t) \label{eq:Sn-1}\\
 &\rightarrow \sum_{\omega=0}^1 \intin\left\{\omega-\widehat{Q}(\beta_0,t) \right\} \E\left\{\de M^{(\omega)}(t)\right\}=0.
\end{align}

\section{Proof for Theorem 1}


Taylor expansion of $G_{n}(\widehat{\beta})=0$ around $\beta_{0}$
leads to
\[
0=G_{n}(\widehat{\beta})=G_{n}(\beta_{0})+\frac{\de }{\de\beta}G_{n}(\tilde{\beta})(\widehat{\beta}-\beta_{0}),
\]
where $\tilde{\beta}$ is on the line segment between $\widehat{\beta}$
and $\beta_{0}$. Then, 
\begin{equation}
n^{1/2}(\widehat{\beta}-\beta_{0})=\left\{ n^{-1}\frac{\de G_{n}(\tilde{\beta})}{\de \beta}\right\} ^{-1}n^{-1/2}G_{n}(\beta_{0}).\label{eq:hat beta}
\end{equation}
Under Assumption \ref{asp:positivity} (ix), $n^{-1}\frac{\de G_{n}(\tilde{\beta})}{\de\beta}>0$ for $\tilde{\beta} \in \mathcal{B}$. Then, the reminder is to show the asymptotic distribution of $n^{-1/2}G_{n}(\beta_{0})$. 
\begin{theorem}\label{Thm:A}Suppose Assumptions \ref{asump:consistency}\textendash \ref{asp:positivity} and Assumption \ref{asump:regua}
hold and $X$ is a continuous scalar variable. Then,
\begin{equation}
n^{-1/2}G_{n}(\beta_{0})\rightarrow\N(0,\overline{V}_{G}),\label{eq:asymp S}
\end{equation}
in distribution, as $n\rightarrow\infty$, where
\begin{align}
\overline{V}_{G}=\sum_{\omega=0}^{1} E \left[\sigma_{H}^{2}(\omega,X)\left\{ \frac{2M+1}{2M}\frac{1}{p(\omega\mid X)}-\frac{1}{2M}p(\omega\mid X)\right\} \right]  + \; \E\left[\left\{ \mu_{H}(0,X)+\mu_{H}(1,X)\right\}^2 \right] .\label{eq:V_S}
\end{align}

\end{theorem}
\begin{proof}

We will show that $n^{-1/2}G_{n}(\beta_{0})$ can be expressed as
a sum of $n$ independent and identically distributed random vectors
plus a term that converges in probability to a zero vector. By some
algebra, we obtain
\begin{equation}
\sum_{i=1}^{n}\{1+k_{i}/M\}\left\{ W_{i}-\widehat{Q}(\beta_{0},t)\right\} \de\Lambda_{0}(t)\exp(\beta_{0}W_{i})Y_{i}(t)=0.\label{eq:(2.3)}
\end{equation}
Therefore, continuing with (\ref{eq:Sn-1}), we obtain
\begin{align}
& n^{-1/2}G_{n}(\beta_{0}) \nonumber \\& =  n^{-1/2}\sum_{i=1}^{n}\{1+k_{i}/M\}\int_{0}^{\tau}\left\{ W_{i}-\widehat{Q}(\beta_{0},t)\right\} \de M_{i}(t) \\
 & =  n^{-1/2}\sum_{i=1}^{n}\{1+k_{i}/M\}\int_{0}^{\tau}\left\{ W_{i}-\widehat{Q}(\beta_{0},t)\right\} \left\{ \de N_{i}(t)-\de\Lambda_{0}(t)\exp(\beta_{0}W_{i})Y_{i}(t)\right\} \label{eq:(2.4)}\\
 & =  n^{-1/2}\sum_{i=1}^{n}\{1+k_{i}/M\}\int_{0}^{\tau}\left\{ W_{i}-Q(\beta_{0},t)\right\} \left\{ \de N_{i}(t)-\de\Lambda_{0}(t) e^{\beta_{0}W_{i}}Y_{i}(t)\right\} \nonumber \\
 & \quad  +n^{-1/2}\sum_{i=1}^{n}\{1+k_{i}/M\}\int_{0}^{\tau}\left\{ Q(\beta_{0},t)-\widehat{Q}(\beta_{0},t)\right\} \left\{ \de N_{i}(t)-\de\Lambda_{0}(t) e^{\beta_{0}W_{i}}Y_{i}(t)\right\} ,\label{eq:(3.2)}
\end{align}
where (\ref{eq:(2.4)}) follows from (\ref{eq:(2.3)}). Under Assumptions S1 (viii), (ix) and (xi), there exists a function $Q(\beta_0,t)$ such that,
\[
\int_{0}^{\tau}\left\{ \widehat{Q}(\beta_{0},t)-Q(\beta_{0},t)\right\} \left[ n^{-1/2}\sum_{i=1}^{n}(1+k_{i}/M)
\left\{ \de N_{i}(t)-\de\Lambda_{0}(t) e^{\beta_{0}W_{i}}Y_{i}(t)\right\}
\right] \rightarrow0,
\]
in probability, as $n\rightarrow\infty$. Therefore, (\ref{eq:(3.2)}) becomes
\begin{equation}
n^{-1/2}G_{n}(\beta_{0})=n^{-1/2}\sum_{i=1}^{n}\{1+k_{i}/M\}H_{i}(W_{i})+o_{p}(1),\label{eq:s0}
\end{equation}
where $H_{i}(\omega)$ is defined in (\ref{eq:def H(w)}). Moreover,
\begin{align*}
\E\{H_{i}(W_{i})\} & =  \int_{0}^{\tau}\E\left[\left\{ W_{i}-Q(\beta_{0},t)\right\} \left\{ \de N_{i}(t)-\de\Lambda_{0}(t)\exp(\beta_{0}W_{i})Y_{i}(t)\right\} \right]\\
 & =  \int_{0}^{\tau}\E\left[e(X_{i})\left\{ 1-Q(\beta_{0},t)\right\} \de M^{(1)}(t)\right] \\
 & \; \quad -\int_{0}^{\tau}\E\left[\left\{ 1-e(X_{i})\right\} Q(\beta_{0},t)\de M^{(0)}(t)\right]\\
 & =  0,
\end{align*}
where the last line follows by the martingale property for the potential
outcome process. 

We write
\begin{align*}
& n^{-1/2}G_{n}(\beta_{0})\\ & =  n^{-1/2}\left\{ G_{n}(\beta_{0})-\sum_{\omega=0}^{1}\sum_{i=1}^{n}\mu_{H}(\omega,X_{i})\right\} +n^{-1/2}\sum_{\omega=0}^{1}\sum_{i=1}^{n}\mu_{H}(\omega,X_{i})\\
 & =  n^{-1/2}\left[\sum_{i=1}^{n}\{1+k_{i}/M\}H_{i}(W_{i})-\sum_{i=1}^{n}\sum_{\omega=0}^{1}\mu_{H}(\omega,X_{i})\right]+n^{-1/2}\sum_{\omega=0}^{1}\sum_{i=1}^{n}\mu_{H}(\omega,X_{i})\\
 & =  \sum_{\omega=0}^{1}n^{-1/2}\sum_{i=1}^{n}\Bigg[I(W_{i}=\omega)\{1+k_{i}/M\}\{H_{i}(\omega)-\mu_{H}(\omega,X_{i})\}\\
 &  \qquad \qquad \qquad \qquad +\{1-I(W_{i}=\omega)\}\{\mu_{H}(\omega,X_{m(\omega,X_{i})})-\mu_{H}(\omega,X_{i})\} +\mu_{H}(\omega,X_{i})\Bigg].
\end{align*}
Similar to (\ref{eq:T2}), we{} have $n^{-1/2}\sum_{i=1}^{n}\{1-I(W_{i}=\omega)\}\{\mu_{H}(\omega,X_{m(\omega,X_{i})})-\mu_{H}(\omega,X_{i})\}=O_{p}(n^{-1/2})=o_{p}(1)$,
for $\omega=0,1$. Therefore, we can write
\[
n^{-1/2}G_{n}(\beta_{0})=\sum_{l=1}^{2n}\xi_{n,l}+o_{p}(1),
\]
where 
\[
\xi_{n,l}=\begin{cases}
n^{-1/2}\left\{ \mu_{H}(0,X_{l})+\mu_{H}(1,X_{l})\right\} , & 1\leq l\leq n,\\
n^{-1/2}\left\{ 1+k_{l-n}\right\} \left\{ H_{l-n}(W_{l-n})-\mu_{H}(W_{l-n},X_{l-n})\right\} , & n+1\leq l\leq2n.
\end{cases}
\]
Consider the $\sigma$-fields 
\[
\mathcal{F}_{n,l}=\begin{cases}
\left\{ W_{1},\ldots,W_{l},X_{1},\ldots,X_{l}\right\} , & 1\leq l\leq n,\\
\left\{ W_{1},\ldots,W_{n},X_{1},\ldots,X_{n},H_{1}(W_{1}),\ldots,H_{l-n}(W_{l-n})\right\} , & n+1\leq l\leq2n.
\end{cases}
\]
Then, for each $n\geq1$, 
\[
\left\{ \sum_{j=1}^{l}\xi_{n,j},\mathcal{F}_{n,l},1\leq l\leq2n\right\} 
\]
is a martingale. Moreover, we evaluate 
\[
\sum_{l=1}^{n}E\left(\xi_{n,l}^{2}\mid\mathcal{F}_{n,l-1}\right)\rightarrow 
\E\left[\left\{ \mu_{H}(0,X)+\mu_{H}(1,X)\right\}^2 \right],
\]
and based on Lemma \ref{lemma A.11}
\begin{align*}
\sum_{l=n+1}^{2n}E\left(\xi_{n,l}^{2}|\mathcal{F}_{n,l-1}\right) & =  n^{-1}\sum_{i=1}^{n}\sum_{\omega=0}^{1}I(W_{i}=\omega)\left\{ 1+k_{i}/M\right\} ^{2}\sigma_{H}^{2}(\omega,X_{i})\\
 & \rightarrow  \sum_{\omega=0}^{1}\E\left[\sigma_{H}^{2}(\omega,X)\left\{ \frac{2M+1}{2M}\frac{1}{p(\omega\mid X)}-\frac{1}{2M}p(\omega\mid X)\right\} \right],
\end{align*}
 as $n\rightarrow\infty$. Apply the Central Limit Theorem for martingale
arrays, (\ref{eq:asymp S}) follows. 

\end{proof}
To establish the result in Theorem \ref{Thm1}, we replace $X$ by
$e(X)$ as the matching variable; therefore, (\ref{eq:asymp S}) holds
for
\begin{align}
V_{G}&=\sum_{\omega=0}^{1}\E\left[\sigma_{H}^{2}\{\omega,e(X)\}\left\{ \frac{2M+1}{2M}\frac{1}{p(\omega\mid X)}-\frac{1}{2M}p(\omega\mid X)\right\} \right] \nonumber \\ 
& \quad + E\left(\left[\mu_{H}\{0,e(X)\} + \mu_{H}\{1,e(X)\}\right]^2\right).
\label{eq:VSforE}
\end{align}

Combining (\ref{eq:hat beta}) and (\ref{eq:VSforE}), $n^{1/2}(\widehat{\beta}-\beta_{0})$
is asymptotically normal with mean zero and covariance matrix $V_{1}=A(\beta_{0})^{-1}V_{G}A(\beta_{0})^{-1}$.
This completes the proof for Theorem \ref{Thm1}.

\section{Proof for Theorem 2}
\begin{theorem}\label{Thm:E}Suppose that Assumptions \ref{asump:consistency}\textendash \ref{asp:positivity} and Assumption \ref{asump:regua}
hold. Suppose that $e(X)$ follows a logistic regression model $e(X^{\T}\theta)$
with the true parameter value $\theta_{0}$. Let $\widehat{\theta}$ be
the maximum likelihood estimator for $\theta$, and $\mathcal{I}_{\theta_{0}}$
be the Fisher information matrix. Then, based on matching on the estimated
propensity score $e(X^{\T}\widehat{\theta})$, 
\[
n^{-1/2}G_{n}(\beta_{0})\rightarrow\N\left(0,V_{G}-c^{\T}\mathcal{I}_{\theta_{0}}^{-1}c\right),
\]
where
\[
c=\E\left\{ \left[\frac{\cov\left\{ X,\mu_{H}(1,X)\mid e(X)\right\} }{e(X)}+\frac{\cov\left\{ X,\mu_{H}(0,X)\mid e(X)\right\} }{1-e(X)}\right]\dot{e}(X^{\T}\theta_{0})\right\} .
\]

\end{theorem}

\begin{proof}Let $Z_{i}=\left\{ X_{i},W_{i},H_{i}(W_{i})\right\} $,
and let $L\left(\theta\mid Z_{1},\ldots,Z_{n}\right)$ be the log
likelihood function of $\theta$, i.e., 
\begin{align*}
L(\theta\mid Z_{1},\ldots,Z_{n})&=\log\left[\prod_{i=1}^{n}e(X_{i}^{\T}\theta)^{W_{i}}\{1-e(X_{i}^{\T}\theta)\}^{1-W_{i}}\right] \\&=\sum_{i=1}^{n}\left[W_{i}\log e(X_{i}^{\T}\theta)+(1-W_{i})\log\{1-e(X_{i}^{\T}\theta)\}\right].
\end{align*}
Following \cite{abadie2016matching} we use the local experiment
argument. Let $\theta_{n}=\theta_{0}+n^{-1/2}h$, and $P^{\theta_{n}}$
be the data distribution under $e(X^{\T}\theta_{n})$. Also, we define
\begin{equation}
\Lambda_{n}(\theta_{0}\mid\theta_{n})=L(\theta_{0}\mid Z_{1},\ldots,Z_{n})-L(\theta_{n}\mid Z_{1},\ldots,Z_{n}).\label{eq:3}
\end{equation}
Under $P^{\theta_{n}},$ we can express $n^{-1/2}G_{n}(\beta_{0})=D_{n}(\theta_{n})+o_{P}(1),$
where
\begin{align*}
D_{n}(\theta_{n})=n^{-1/2}\sum_{i=1}^{n}\sum_{\omega=0}^{1}\bigg[&I(W_{i}=\omega)\left\{ 1+k_{i}/M\right\} \left[H_{i}(\omega)-\mu_{H}\left\{ \omega,e(X_{i}^{\T}\theta_{n})\right\} \right]\\& \;+\mu_{H}\left\{ \omega,e(X_{i}^{\T}\theta_{n})\right\} \bigg].
\end{align*}

We shall show that under $P^{\theta_{n}}:$
\begin{equation}
\left(\begin{array}{c}
D_{n}(\theta_{n})\\
n^{1/2}(\widehat{\theta}_{n}-\theta_{n})\\
\Lambda_{n}(\theta_{0}\mid\theta_{n})
\end{array}\right)\rightarrow\N\left\{ \left(\begin{array}{c}
0\\
0\\
-h^{\T}\mathcal{I}_{\theta_{0}}h/2
\end{array}\right),\left(\begin{array}{ccc}
V_{G} & c^{\T}\mathcal{I}_{\theta_{0}}^{-1} & -c^{\T}h\\
\mathcal{I}_{\theta_{0}}^{-1}c & \mathcal{I}_{\theta_{0}}^{-1} & h\\
-h^{\T}c & -h^{\T} & h^{\T}\mathcal{I}_{\theta_{0}}h
\end{array}\right)\right\} ,\label{eq:lecam}
\end{equation}
in distribution, as $n\rightarrow\infty$. Then, by Le Cam's third
lemma \citep{le1990asymptotics}, $n^{-1/2}G_{n}(\beta_{0})\rightarrow\N\left(0,V_{G}-c^{\T}\mathcal{I}_{\theta_{0}}^{-1}c\right)$
in distribution, as $n\rightarrow\infty$. 

To show (\ref{eq:lecam}), denote \textcolor{black}{
\[
\Delta_{n}(\theta){\color{black}=}n^{-1/2}\frac{\partial}{\partial\theta}L(\theta\mid Z_{1},\ldots,Z_{n})=n^{-1/2}\sum_{i=1}^{n}X_{i}\dot{e}(X_{i}^{\T}\theta)\frac{W_{i}-e(X_{i}^{\T}\theta)}{e(X_{i}^{\T}\theta)\{1-e(X_{i}^{\T}\theta)\}}.
\]
Then, under $P^{\theta_{n}}$:}
\[
\Delta_{n}(\theta_{n})\rightarrow\N(0,\mathcal{I}_{\theta_{0}}),\ \mathcal{I}_{\theta_{0}}=\E\left[\frac{\dot{e}(X^{\T}\theta)^{2}XX^{\T}}{e(X^{\T}\theta)\{1-e(X^{\T}\theta)\}}\right],
\]
in distribution, as $n\rightarrow\infty$.\textcolor{black}{We also
note that under $P^{\theta_{n}}$:}
\begin{eqnarray}
n^{1/2}(\widehat{\theta}_{n}-\theta_{n}) & = & \mathcal{I}_{\theta_{0}}^{-1}\Delta_{n}(\theta_{n})+o_{P}(1),\label{eq:4}\\
\Lambda_{n}(\theta_{0}\mid\theta_{n}) & = & -h^{\T}\Delta_{n}(\theta_{n})-\frac{1}{2}h^{\T}\mathcal{I}_{\theta_{0}}h+o_{P}(1).\nonumber 
\end{eqnarray}
To show (\ref{eq:lecam}), it suffices to show that
\begin{equation}
\left(\begin{array}{c}
D_{n}(\theta_{n})\\
\Delta_{n}(\theta_{n})
\end{array}\right)\rightarrow\N\left\{ \left(\begin{array}{c}
0\\
0
\end{array}\right),\left(\begin{array}{cc}
V_{G} & c^{\T}\\
c & \mathcal{I}_{\theta_{0}}
\end{array}\right)\right\} ,\label{eq:5}
\end{equation}
in distribution, as $n\rightarrow\infty$. Toward this end, we consider
the linear combination $L_{n}=z_{1}D_{n}(\theta_{n})+z_{2}^{\T}\Delta_{n}(\theta_{n}),$
for any $z_{1}$ and $z_{2}$. We write $L_{n}=\sum_{l=1}^{3n}\xi_{n,l}$,
where \textcolor{black}{{} 
\begin{align*}
\xi_{n,l}=\begin{cases}
&z_{1}n^{-1/2}\sum_{\omega=0}^{1}\mu_{H}\left\{ \omega,e(X_{l}^{\T}\theta_{n})\right\} \\&+z_{2}^{\T}{\color{black}n^{-1/2}\E\left\{ X_{l}\mid e(X_{l}^{\T}\theta_{n})\right\} \dot{e}(X_{i}^{\T}\theta_{n})\frac{W_{i}-e(X_{i}^{\T}\theta_{n})}{e(X_{i}^{\T}\theta_{n})\{1-e(X_{i}^{\T}\theta_{n})\}},} 
\\&\qquad \qquad \qquad \qquad \qquad
\qquad \qquad \qquad \qquad \qquad 
\qquad \qquad \quad
\text{for \;} 1\leq l\leq n;\\
&z_{1}n^{-1/2}\sum_{\omega=0}^{1}I(W_{l-n}=\omega)\left\{ 1+k_{l-n}/M\right\} \left[\mu_{H}(\omega,X_{l-n})-\mu_{H}\left\{ \omega,e(X_{l-n}^{\T}\theta_{n})\right\} \right]\\
&+z_{2}^{\T}n^{-1/2}\left(X_{l-n}-\E\left\{ X_{l-n}\mid e(X_{l-n}^{\T}\theta_{n})\right\} \right)\dot{e}(X_{l-n}^{\T}\theta_{n})\frac{W_{l-n}-e(X_{l-n}^{\T}\theta_{n})}{e(X_{l-n}^{\T}\theta_{n})\{1-e(X_{l-n}^{\T}\theta_{n})\}},\\  &\qquad \qquad \qquad \qquad \qquad 
\qquad \qquad \qquad \qquad \qquad 
\qquad \qquad \quad
\text{for \;} n+1\leq l\leq2n;\\
&z_{1}n^{-1/2}\sum_{\omega=0}^{1}I(W_{l-2n}=\omega)\left\{ 1+k_{l-2n}/M\right\} \left\{ H_{l-2n}(W_{l-2n})-\mu_{H}(\omega,X_{l-2n})\right\} , \\&\qquad \qquad \qquad \qquad \qquad
\qquad \qquad \qquad \qquad \qquad 
\qquad \qquad \quad
\text{for \;} 2n+1\leq l\leq3n.
\end{cases}
\end{align*}
}Consider $\sigma-$fields\textcolor{black}{
\[
\mathcal{F}_{n,l}=\begin{cases}
\left\{ W_{1},\ldots,W_{l},X_{1}^{\T}\theta_{n},\ldots,X_{l}^{\T}\theta_{n}\right\} , & 1\leq l\leq n,\\
\left\{ W_{1},\ldots,W_{n},X_{1}^{\T}\theta_{n},\ldots,X_{n}^{\T}\theta_{n},X_{1},\ldots,X_{l-n}\right\} , & n+1\leq l\leq2n,\\
\left\{ W_{1},\ldots,W_{n},X_{1},\ldots,X_{n},H_{1}(W_{1}),\ldots,H_{l-2n}(W_{l-2n})\right\} , & 2n+1\leq l\leq3n.
\end{cases}
\]
}Then, $\left\{ \sum_{j=1}^{l}\xi_{n,j},F_{n,i},1\leq l\leq3n\right\} $
is a martingale for each $n\geq1$. Under $P^{\theta_{n}}$, 
\[
L_{n}\rightarrow\N(0,\sigma_{L,1}^{2}+\sigma_{L,2}^{2}+\sigma_{L,3}^{2}),
\]
in distribution, as $n\rightarrow\infty,$ where
\begin{align*}
\sigma_{L,1}^{2}&=z_{1}^{2}E\left(\left[\mu_{H}\left\{ 0,e(X)\right\} +\mu_{H}\left\{ 1,e(X)\right\}\right]^2\right)\\
& \; \quad +z_{2}^{\T}\E\left[\E\left\{ X\mid e(X)\right\} \mbox{E}\left\{ X^{\T}\mid e(X)\right\} \frac{\dot{e}(X^{\T}\theta_{0})^{2}}{e(X)\{1-e(X)\}}\right]z_{2},
\end{align*}
\begin{align*}
& \sigma_{L,2}^{2} \\  
& =  z_{1}^{2}\sum_{\omega=0}^{1}\E\left[\var\left\{ \mu_{H}(\omega,X)\mid e(X)\right\} \left\{ \frac{2M+1}{2M}\frac{1}{p(\omega\mid X^{\T}\theta_{0})}-\frac{1}{2M}p(\omega\mid X^{\T}\theta_{0})\right\} \right]
\\
 & \quad \; + \; 2z_{2}^{\T}\E\left[\frac{\cov\left\{ X,\mu_{H}(1,X)\mid e(X)\right\} \dot{e}(X^{\T}\theta_{0})}{e(X)}-\frac{\cov\left\{ X,\mu_{H}(0,X)\mid e(X)\right\} \dot{e}(X^{\T}\theta_{0})}{1-e(X)}\right]z_{1}\\
 & \quad \; + \; z_{2}^{\T}\E\left[\var\left\{ X\mid e(X)\right\} \frac{\dot{e}(X^{\T}\theta_{0})^{2}}{e(X)\{1-e(X)\}}\right]z_{2},
\end{align*}
and
\[
\sigma_{L,3}^{2}=z_{1}^{2}\sum_{\omega=0}^{1} \E\left[\sigma_{H}^{2}(\omega,X)\mid e(X)\right]\left\{ \frac{2M+1}{2M}\frac{1}{p(\omega\mid X)}+\frac{1}{2M}p(\omega\mid X)\right\} .
\]
Then, $\sigma_{L,1}^{2}+\sigma_{L,2}^{2}+\sigma_{L,3}^{2}=z_{1}^{2}V_{G}+z_{2}^{\T}\mathcal{I}_{\theta_{0}}^{-1}z_{2}+2z_{2}^{\T}cz_{1}.$
Thus, under $P^{\theta_{n}}$, (\ref{eq:5}) follows. This completes
the proof for Theorem \ref{Thm2}. 

\end{proof}

\section{Simulation}

\subsection{An data-generating algorithm}

We describe the algorithm for generating $T^{(0)}$ and $T^{(1)}$ that are congenial
with Model (\ref{(5.1)}) as follows.

\begin{algorithm}[hbt!] 
%
\vspace{6pt}
\begin{tabbing}
\qquad \enspace {Step$\ 1.$} Generate $T^{(0)}$ from $S^{(0)}(t)=\exp(-\lambda_{0}t)$,
where $\lambda_{0}=6$ for $\beta_0=0$ and 0.5, and \\ 
\qquad \qquad $\lambda_{0}=15$ for $\beta_0=-0.5$. \\
\qquad \enspace {Step$\ 2.$} Generate $u$ from Unif$[0,1]$, solve \\
\qquad \qquad \qquad $\left\{ \prod_{k=1}^{6}\left(1-\frac{\eta_{k}t}{\lambda_{k}}\right)\right\} \exp\left[\left\{ X^{\T}\eta-\lambda_{0}\exp(\beta_{0})\right\} t\right]-1+u=0$ \\
\qquad \qquad for $t$, where $\eta_{1}=\cdots=\eta_{6}=-2$ and $\eta=\left(\eta_1,\cdots,\eta_6\right)^{\T}$. Let $T^{(1)}$ be the solution $t$. 
\end{tabbing}
\vspace{6pt}
 \caption{for generating $T^{(0)}$ and $T^{(1)}$ that are congenial with Model (\ref{(5.1)}) \label{alg:1}}

\end{algorithm}

By Algorithm \ref{alg:1}, $T^{(1)}$ given $X$ follows 
\[
S_{T\mid W,X}(t\mid W=1,X)=\left\{ \prod_{k=1}^{6}\left(1-\frac{\eta_{k}t}{\lambda_{k}}\right)\right\} \exp\left[\left\{ X^{\T}\mathbf{\eta }-\lambda_{0}\exp(\beta_{0})\right\} t\right].
\]
Under our parameter specification, we have (i) $S_{T\mid W,X}(t=0\mid W=1,X=x)=1$,
(ii) $S_{T\mid W,X}(t=\tau\mid W=1,X=x)=0$, because of $X^{\T}\mathbf{\eta }-\lambda_{0}\exp(\beta_{0})\leq 0$
, and (iii) $\de S_{T\mid W,X}(t\mid W=1,X=x)/\de t\leq0$, because
of $\sum_{i=1}^6\lambda_{i}^{-1}\leq\min(\lambda_{0},\lambda_{0}e^{\beta_{0}})$.

The marginal distribution of $T^{(1)}$ is 
\begin{eqnarray*}
S^{(1)}(t) & = & \int S_{T\mid W,X}(t\mid W=1,X=x)\de F(x)\\
 & = & \int\left\{ \prod_{k=1}^{6}\left(1-\frac{\eta_{k}t}{\lambda_{k}}\right)\right\} \exp\left[\left\{ X^{\T}\mathbf{\eta }-\lambda_{0}\exp(\beta_{0})\right\} t\right]\de F(x)\\
 & = & \left\{ S^{(0)}(t)\right\} ^{\exp(\beta_{0})}.
\end{eqnarray*}
Therefore, the marginal distributions of $\{T^{(0)},T^{(1)}\}$ satisfy
$S^{(1)}(t)=\left\{ S^{(0)}(t)\right\} ^{\exp(\beta_{0})}$ . 

\subsection{Illustration of propensity score distributions with weak, medium,
and strong covariate overlap}

This section demonstrates the mentioned simulation settings of weak, medium, strong covariate
overlap between
the two treatment groups.
Figure \ref{fig:ps-separations} shows density curves for the true propensity scores of
the two treatment groups are presented 
in dashed lines of $W=0$ and in solid lines of $W=1$.
\begin{figure}[ht]
\begin{centering}
\includegraphics[width=\textwidth]{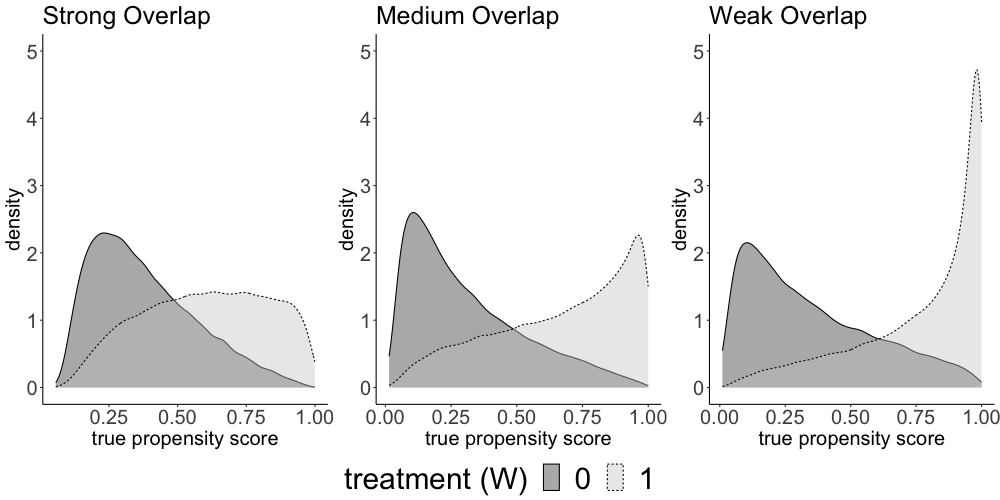}
\par\end{centering}
\caption{\label{fig:ps-separations} Plots of the true propensity score distribution
between treatment groups $W=0$ and $W=1$ for strong, medium and weak covariate overlap. }

\end{figure}
\subsection{Additional simulation results}

We provide additional simulation results in Table \ref{tab:beta0-N1000}-\ref{tab:ps=0.5}. 

\begin{table}[hbt!]
\caption{\label{tab:beta0-N1000}Simulation results: bias ($\times10^{2}$)
and variance ($\times10^{3}$) of the point estimator of $\beta_{0}$,
coverage ($\%$) of $95\%$ confidence intervals based on $1,000$
Monte Carlo samples with true $\beta_{0}=0$ when propensity score model is misspecified.} 
\resizebox{\textwidth}{!}{%
\begin{tabular}{llllllllllllll} 
\hline
\multicolumn{2}{l}{}                                            & Bias & Var   & VE   & CR     & Bias & Var   & VE    & CR     & Bias & Var    & VE    & CR      \\
\multicolumn{2}{l}{Level of covariate overlap}                           & \multicolumn{4}{l}{Strong}     & \multicolumn{4}{l}{Medium}    & \multicolumn{4}{l}{Weak}      \\ 
\hline
\multicolumn{14}{l}{Scenario (i) Correct specification of the propensity score model}                                                                            \\
$\widehat{\beta}_{\nai}$                       &                & 54.6 & 4.9   & 5.1  & 0.0\%  & 70.6 & 5.2   & 5.3   & 0.0\%  & 80.9 & 5.2    & 5.7   & 0.0\%   \\
$\widehat{\beta}_{\ipw}$                       &                & 0.3  & 6.1   & 8.3  & 97.8\% & 1.0  & 12.5  & 14.6  & 95.6\% & 4.1  & 22.2   & 18.9  & 91.2\%  \\
$\widehat{\beta}_{\aipw}$                      &                & 0.1  & 5.3   & 6.0  & 95.9\% & 0.2  & 9.1   & 12.0  & 96.1\% & 2.3  & 31.3   & 40.6  & 96.5\%  \\
$\widehat{\beta}_{\xm}$                        &                & 6.9  & 5.9   & 8.8  & 92.7\% & 11.4 & 8.2   & 11.8  & 84.0\% & 17.7 & 8.2    & 11.7  & 64.4\%  \\
$\widehat{\beta}_{\psm.0}$                     &                & 0.5  & 6.6   & 10.9 & 98.5\% & 1.6  & 10.6  & 17.9  & 97.4\% & 3.7  & 16.1   & 25.4  & 96.1\%  \\
\multirow{5}{*}{\begin{tabular}[c]{@{}l@{}}$\widehat{\beta}_{\psm}$\\ (M=1)\end{tabular}} & software       & 0.7  & 6.3   & 11.1 & 98.9\% & 1.3  & 10.5  & 18.3  & 97.9\% & 3.9  & 17.2   & 26.0  & 95.8\%  \\
                                               & asymp          & 0.7  & 6.3   & 6.5  & 95.1\% & 1.3  & 10.5  & 10.6  & 94.2\% & 3.9  & 17.2   & 16.4  & 91.3\%  \\
                                               & naiveboot      & 0.7  & 6.3   & 9.8  & 98.8\% & 1.3  & 10.5  & 17.0  & 98.7\% & 3.9  & 17.2   & 25.1  & 96.3\%  \\
                                               & double-rsp(5)  & 0.7  & 6.3   & 8.2  & 97.7\% & 1.3  & 10.5  & 13.6  & 97.5\% & 3.9  & 17.2   & 24.1  & 95.4\%  \\
                                               & double-rsp(10) & 0.7  & 6.3   & 8.9  & 98.2\% & 1.3  & 10.5  & 14.7  & 97.4\% & 3.9  & 17.2   & 22.2  & 94.5\%  \\
\multirow{5}{*}{\begin{tabular}[c]{@{}l@{}}$\widehat{\beta}_{\psm}$\\ (M=5)\end{tabular}}  & software       & 1.5  & 5.0   & 8.1  & 99.0\% & 3.2  & 8.0   & 12.3  & 96.6\% & 5.1  & 10.1   & 16.3  & 97.0\%  \\
                                               & asymp          & 1.5  & 5.0   & 5.0  & 94.2\% & 3.2  & 8.0   & 7.3   & 91.9\% & 5.1  & 10.1   & 9.2   & 90.6\%  \\
                                               & naiveboot      & 1.5  & 5.0   & 7.6  & 99.0\% & 3.2  & 8.0   & 12.2  & 96.8\% & 5.1  & 10.1   & 16.6  & 97.5\%  \\
                                               & double-rsp(5)  & 1.5  & 5.0   & 6.1  & 96.1\% & 3.2  & 8.0   & 9.5   & 95.2\% & 5.1  & 10.1   & 13.9  & 95.8\%  \\
                                               & double-rsp(10) & 1.5  & 5.0   & 6.5  & 97.1\% & 3.2  & 8.0   & 10.0  & 96.0\% & 5.1  & 10.1   & 14.4  & 96.2\%  \\ 
\hline
\multicolumn{14}{l}{Scenario (ii) Misspecification of the propensity score model}                                                                                \\
$\widehat{\beta}_{\nai}$                       &                & 54.6 & 4.9   & 5.1  & 0.0\%  & 70.4 & 5.0   & 5.3   & 0.0\%  & 80.8 & 5.2    & 5.7   & 0.0\%   \\
$\widehat{\beta}_{\ipw}$                       &                & 4.5  & 19.4  & 12.4 & 84.6\% & 3.8  & 42.2  & 19.6  & 79.0\% & 3.4  & 68.5   & 26.6  & 74.7\%  \\
$\widehat{\beta}_{\aipw}$                      &                & 17.2 & 121.5 & 72.8 & 73.1\% & 30.1 & 614.2 & 308.9 & 79.7\% & 30.0 & 1054.2 & 603.7 & 89.4\%  \\
$\widehat{\beta}_{\xm}$                        &                & 6.9  & 5.9   & 8.8  & 92.7\% & 11.4 & 8.3   & 11.8  & 84.1\% & 17.6 & 8.1    & 11.7  & 64.8\%  \\
$\widehat{\beta}_{\psm.0}$                     &                & 0.5  & 6.6   & 10.9 & 98.5\% & 1.3  & 10.7  & 18.0  & 97.4\% & 3.3  & 15.6   & 25.4  & 96.5\%  \\
\multirow{5}{*}{\begin{tabular}[c]{@{}l@{}}$\widehat{\beta}_{\psm}$\\ (M=1)\end{tabular}}  & software       & 5.4  & 6.7   & 10.6 & 95.3\% & 7.7  & 10.4  & 16.8  & 93.7\% & 10.9 & 15.1   & 22.8  & 89.6\%  \\
                                               & asymp          & 5.4  & 6.7   & 6.5  & 88.3\% & 7.7  & 10.4  & 9.4   & 86.1\% & 10.9 & 15.1   & 11.0  & 77.4\%  \\
                                               & naiveboot      & 5.4  & 6.7   & 9.5  & 94.6\% & 7.7  & 10.4  & 15.8  & 93.9\% & 10.9 & 15.1   & 21.6  & 90.9\%  \\
                                               & double-rsp(5)  & 5.4  & 6.7   & 9.8  & 95.5\% & 7.7  & 10.4  & 16.1  & 93.4\% & 10.9 & 15.1   & 17.7  & 88.2\%  \\
                                               & double-rsp(10) & 5.4  & 6.7   & 11.0 & 95.6\% & 7.7  & 10.4  & 21.2  & 94.5\% & 10.9 & 15.1   & 19.8  & 89.6\%  \\
\multirow{5}{*}{\begin{tabular}[c]{@{}l@{}}$\widehat{\beta}_{\psm}$\\ (M=5)\end{tabular}}  & software       & 6.0  & 5.1   & 7.9  & 93.9\% & 9.5  & 7.9   & 11.5  & 88.5\% & 13.1 & 9.9    & 14.7  & 84.9\%  \\
                                               & asymp          & 6.0  & 5.1   & 4.8  & 85.0\% & 9.5  & 7.9   & 6.6   & 76.1\% & 13.1 & 9.9    & 5.1   & 64.3\%  \\
                                               & naiveboot      & 6.0  & 5.1   & 7.5  & 93.6\% & 9.5  & 7.9   & 11.6  & 88.9\% & 13.1 & 9.9    & 15.0  & 86.4\%  \\
                                               & double-rsp(5)  & 6.0  & 5.1   & 7.3  & 93.6\% & 9.5  & 7.9   & 12.1  & 89.1\% & 13.1 & 9.9    & 11.8  & 79.1\%  \\
                                               & double-rsp(10) & 6.0  & 5.1   & 8.2  & 94.2\% & 9.5  & 7.9   & 14.4  & 91.4\% & 13.1 & 9.9    & 13.9  & 82.5\%  \\
\hline
\end{tabular}}
Note: "Var" is the variance of point estimates of $\beta_{0}$ across simulated datasets;  "VE" is the average variance estimation for the point estimators over simulations, thus VE minus Var reflects the bias in estimated variance;  "CR" is the empirical coverage rate of $95\%$ confidence intervals. Five types of variance estimates for $\widehat{\beta}_{\psm}$ were compared: "software", output from the standard software; "asymp", the proposed asymptotic variance estimation; "naiveboot", the naive nonparametric bootstrap;
"double-rsp(5)", the proposed double-resampling method with five quantile strata and "double-rsp(10)", the proposed double-resampling method with ten quantile strata.
\end{table}

\begin{table}[hbt!]
\caption{\label{tab:beta-0.5-N1000}Simulation results: bias ($\times10^{2}$)
and variance ($\times10^{3}$) of the point estimator of $\beta_{0}$,
coverage ($\%$) of $95\%$ confidence intervals based on $1,000$
Monte Carlo samples with true $\beta_{0}=0.5$}
\resizebox{\textwidth}{!}{%
\begin{tabular}{llllllllllllll} 
\hline
\multicolumn{2}{l}{}                                           & Bias & Var   & VE   & CR     & Bias & Var   & VE    & CR     & Bias & Var   & VE    & CR      \\
\multicolumn{2}{l}{Level of covariate overlap}                          & \multicolumn{4}{l}{Strong}     & \multicolumn{4}{l}{Medium}    & \multicolumn{4}{l}{Weak}     \\ 
\hline
\multicolumn{14}{l}{Scenario (i) Correct specification of the propensity score model}                                                                          \\
$\widehat{\beta}_{\nai}$                       &                & 46.4 & 5.5   & 5.6  & 0.0\%  & 59.7 & 5.8   & 5.8   & 0.0\%  & 59.7 & 6.1   & 6.5   & 0.0\%   \\
$\widehat{\beta}_{\ipw}$                       &                & -0.5 & 7.1   & 8.7  & 97.0\% & 0.1  & 14.8  & 15.0  & 94.5\% & 0.1  & 28.4  & 21.5  & 89.8\%  \\
$\widehat{\beta}_{\aipw}$                        &                & -0.7 & 6.2   & 7.4  & 95.2\% & -0.6 & 11.6  & 17.3  & 94.4\% & -0.6 & 35.7  & 42.6  & 97.2\%  \\
$\widehat{\beta}_{\xm}$                        &                & 6.3  & 7.7   & 8.9  & 90.3\% & 10.0 & 10.5  & 11.7  & 83.8\% & 16.8 & 10.4  & 12.4  & 66.0\%  \\
$\widehat{\beta}_{\psm.0}$                     &                & -0.2 & 8.6   & 11.3 & 96.9\% & 0.5  & 14.9  & 18.6  & 95.8\% & 0.5  & 23.5  & 29.9  & 94.0\%  \\
\multirow{5}{*}{\begin{tabular}[c]{@{}l@{}}$\widehat{\beta}_{\psm}$\\ (M=1)\end{tabular}} & software       & -0.1 & 8.3   & 11.5 & 97.0\% & 0.2  & 14.6  & 18.8  & 96.3\% & 0.2  & 25.3  & 30.6  & 93.2\%  \\
                                              & asymp          & -0.1 & 8.3   & 8.4  & 94.5\% & 0.2  & 14.6  & 14.5  & 92.9\% & 0.2  & 25.3  & 23.1  & 90.6\%  \\
                                              & naiveboot      & -0.1 & 8.3   & 11.1 & 97.1\% & 0.2  & 14.6  & 18.6  & 96.8\% & 0.2  & 25.3  & 31.7  & 94.6\%  \\
                                              & double-rsp(5)  & -0.1 & 8.3   & 9.6  & 95.9\% & 0.2  & 14.6  & 16.0  & 95.5\% & 0.2  & 25.3  & 29.9  & 95.0\%  \\
                                              & double-rsp(10) & -0.1 & 8.3   & 10.4 & 96.5\% & 0.2  & 14.6  & 16.9  & 96.1\% & 0.2  & 25.3  & 27.5  & 93.3\%  \\
\multirow{5}{*}{\begin{tabular}[c]{@{}l@{}}$\widehat{\beta}_{\psm}$\\ (M=5)\end{tabular}}& software       & 0.4  & 6.5   & 8.5  & 97.3\% & 1.6  & 10.6  & 12.9  & 95.5\% & 1.6  & 13.7  & 18.9  & 96.5\%  \\
                                              & asymp          & 0.4  & 6.5   & 6.3  & 93.9\% & 1.6  & 10.6  & 9.4   & 92.6\% & 1.6  & 13.7  & 12.5  & 91.9\%  \\
                                              & naiveboot      & 0.4  & 6.5   & 8.9  & 97.9\% & 1.6  & 10.6  & 13.8  & 96.8\% & 1.6  & 13.7  & 20.9  & 97.4\%  \\
                                              & double-rsp(5)  & 0.4  & 6.5   & 7.0  & 95.9\% & 1.6  & 10.6  & 10.9  & 94.8\% & 1.6  & 13.7  & 16.8  & 96.2\%  \\
                                              & double-rsp(10) & 0.4  & 6.5   & 7.5  & 96.5\% & 1.6  & 10.6  & 11.4  & 95.0\% & 1.6  & 13.7  & 17.5  & 96.5\%  \\ 
\hline
\multicolumn{14}{l}{Scenario (ii) Misspecification of the propensity score model}                                                                              \\
$\widehat{\beta}_{\nai}$                       &                & 46.4 & 5.5   & 5.6  & 0.0\%  & 59.7 & 5.6   & 5.8   & 0.0\%  & 72.0 & 6.0   & 6.5   & 0.0\%   \\
$\widehat{\beta}_{\ipw}$                       &                & -0.8 & 27.3  & 15.7 & 91.5\% & -3.5 & 61.3  & 25.5  & 84.7\% & -1.9 & 90.4  & 35.2  & 79.4\%  \\
$\widehat{\beta}_{\aipw}$                        &                & 11.6 & 111.1 & 73.3 & 90.4\% & 20.1 & 607.2 & 308.4 & 93.7\% & 10.5 & 778.5 & 538.8 & 94.4\%  \\
$\widehat{\beta}_{\xm}$                        &                & 6.8  & 8.1   & 9.7  & 90.5\% & 10.4 & 11.7  & 13.4  & 82.9\% & 18.1 & 13.5  & 15.1  & 65.8\%  \\
$\widehat{\beta}_{\psm.0}$                     &                & -0.3 & 8.8   & 11.2 & 96.8\% & 0.3  & 14.4  & 18.6  & 96.3\% & 2.2  & 23.2  & 30.1  & 94.5\%  \\
\multirow{5}{*}{\begin{tabular}[c]{@{}l@{}}$\widehat{\beta}_{\psm}$\\ (M=1)\end{tabular}} & software       & 4.3  & 8.7   & 10.9 & 93.9\% & 6.3  & 13.4  & 17.1  & 92.3\% & 9.7  & 20.7  & 26.2  & 89.2\%  \\
                                              & asymp          & 4.3  & 8.7   & 8.3  & 90.9\% & 6.3  & 13.4  & 12.3  & 88.2\% & 9.7  & 20.7  & 14.5  & 81.1\%  \\
                                              & naiveboot      & 4.3  & 8.7   & 10.7 & 93.9\% & 6.3  & 13.4  & 17.0  & 92.7\% & 9.7  & 20.7  & 26.9  & 91.6\%  \\
                                              & double-rsp(5)  & 4.3  & 8.7   & 11.0 & 94.5\% & 6.3  & 13.4  & 17.7  & 93.9\% & 9.7  & 20.7  & 21.0  & 88.9\%  \\
                                              & double-rsp(10) & 4.3  & 8.7   & 12.1 & 95.2\% & 6.3  & 13.4  & 22.0  & 94.7\% & 9.7  & 20.7  & 23.1  & 90.1\%  \\
\multirow{5}{*}{\begin{tabular}[c]{@{}l@{}}$\widehat{\beta}_{\psm}$\\ (M=5)\end{tabular}}& software       & 4.4  & 6.5   & 8.3  & 94.0\% & 7.4  & 9.9   & 12.0  & 89.7\% & 11.6 & 12.8  & 16.9  & 89.3\%  \\
                                              & asymp          & 4.4  & 6.5   & 6.0  & 89.0\% & 7.4  & 9.9   & 8.1   & 82.0\% & 11.6 & 12.8  & 6.3   & 74.7\%  \\
                                              & naiveboot      & 4.4  & 6.5   & 8.7  & 95.0\% & 7.4  & 9.9   & 12.9  & 91.8\% & 11.6 & 12.8  & 18.6  & 91.1\%  \\
                                              & double-rsp(5)  & 4.4  & 6.5   & 8.1  & 93.4\% & 7.4  & 9.9   & 13.0  & 90.5\% & 11.6 & 12.8  & 13.8  & 84.4\%  \\
                                              & double-rsp(10) & 4.4  & 6.5   & 8.9  & 94.1\% & 7.4  & 9.9   & 14.8  & 91.7\% & 11.6 & 12.8  & 16.1  & 86.0\%  \\
\hline
\end{tabular}}
Note: "Var" is the variance of point estimates of $\beta_{0}$ across simulated datasets;  "VE" is the average variance estimation for the point estimators over simulations, thus VE minus Var reflects the bias in estimated variance;  "CR" is the empirical coverage rate of $95\%$ confidence intervals. Five types of variance estimates for $\widehat{\beta}_{\psm}$ were compared: "software", output from the standard software; "asymp", the proposed asymptotic variance estimation; "naiveboot", the naive nonparametric bootstrap;
"double-rsp(5)", the proposed double-resampling method with five quantile strata and "double-rsp(10)", the proposed double-resampling method with ten quantile strata.
\end{table}

\begin{table}[hbt!]
\caption{\label{tab:beta0.5-N1000}Simulation results: bias ($\times10^{2}$)
and variance ($\times10^{3}$) of the point estimator of $\beta_{0}$,
coverage ($\%$) of $95\%$ confidence intervals based on $1,000$
Monte Carlo samples with true $\beta_{0}=-0.5$}
\resizebox{\textwidth}{!}{%
\begin{tabular}{llllllllllllll} 
\hline
\multicolumn{2}{l}{}                                           & Bias & Var   & VE   & CR     & Bias & Var   & VE    & CR     & Bias & Var    & VE    & CR      \\
\multicolumn{2}{l}{Level of covariate overlap}                          & \multicolumn{4}{l}{Strong}     & \multicolumn{4}{l}{Medium}    & \multicolumn{4}{l}{Weak}      \\ 
\hline
\multicolumn{14}{l}{Scenario (i) Correct specification of the propensity score model}                                                                           \\
$\widehat{\beta}_{\nai}$                       &                & 33.7 & 4.3   & 4.8  & 0.1\%  & 42.8 & 4.6   & 4.8   & 0.0\%  & 44.9 & 4.7    & 5.0   & 0.0\%   \\
$\widehat{\beta}_{\ipw}$                       &                & 0.8  & 6.1   & 7.3  & 97.1\% & 1.0  & 11.5  & 11.6  & 95.3\% & 2.7  & 17.8   & 15.0  & 92.9\%  \\
$\widehat{\beta}_{\aipw}$                        &                & 0.6  & 5.9   & 6.5  & 96.7\% & 0.6  & 10.7  & 11.5  & 94.5\% & 2.6  & 55.3   & 48.1  & 95.7\%  \\
$\widehat{\beta}_{\xm}$                        &                & 4.5  & 6.7   & 8.2  & 94.0\% & 7.3  & 8.9   & 10.3  & 89.4\% & 10.4 & 8.6    & 10.2  & 83.3\%  \\
$\widehat{\beta}_{\psm.0}$                     &                & 0.9  & 8.2   & 9.8  & 96.5\% & 1.7  & 12.8  & 14.9  & 96.5\% & 3.2  & 20.2   & 20.7  & 95.1\%  \\
\multirow{5}{*}{\begin{tabular}[c]{@{}l@{}}$\widehat{\beta}_{\psm}$\\ (M=1)\end{tabular}} & software       & 1.0  & 8.0   & 10.0 & 97.5\% & 1.5  & 13.1  & 15.1  & 96.3\% & 3.3  & 21.5   & 21.1  & 94.9\%  \\
                                              & asymp          & 1.0  & 8.0   & 8.4  & 95.0\% & 1.5  & 13.1  & 12.6  & 93.8\% & 3.3  & 21.5   & 19.2  & 91.9\%  \\
                                              & naiveboot      & 1.0  & 8.0   & 9.3  & 96.6\% & 1.5  & 13.1  & 14.0  & 96.1\% & 3.3  & 21.5   & 19.4  & 94.0\%  \\
                                              & double-rsp(5)  & 1.0  & 8.0   & 9.2  & 95.8\% & 1.5  & 13.1  & 15.0  & 95.5\% & 3.3  & 21.5   & 27.1  & 93.2\%  \\
                                              & double-rsp(10) & 1.0  & 8.0   & 9.8  & 96.9\% & 1.5  & 13.1  & 15.2  & 95.7\% & 3.3  & 21.5   & 25.0  & 92.4\%  \\
\multirow{5}{*}{\begin{tabular}[c]{@{}l@{}}$\widehat{\beta}_{\psm}$\\ (M=5)\end{tabular}} & software       & 1.5  & 6.1   & 7.5  & 96.9\% & 2.5  & 9.8   & 10.9  & 95.5\% & 3.4  & 12.0   & 14.1  & 95.4\%  \\
                                              & asymp          & 1.5  & 6.1   & 6.3  & 95.0\% & 2.5  & 9.8   & 8.8   & 91.6\% & 3.4  & 12.0   & 10.8  & 90.9\%  \\
                                              & naiveboot      & 1.5  & 6.1   & 7.8  & 97.4\% & 2.5  & 9.8   & 11.0  & 95.5\% & 3.4  & 12.0   & 14.2  & 95.3\%  \\
                                              & double-rsp(5)  & 1.5  & 6.1   & 7.0  & 95.7\% & 2.5  & 9.8   & 10.4  & 94.4\% & 3.4  & 12.0   & 16.3  & 95.2\%  \\
                                              & double-rsp(10) & 1.5  & 6.1   & 7.2  & 95.8\% & 2.5  & 9.8   & 10.6  & 95.0\% & 3.4  & 12.0   & 16.2  & 95.6\%  \\ 
\hline
\multicolumn{14}{l}{Scenario (ii) Misspecification of the propensity score model}                                                                               \\
$\widehat{\beta}_{\nai}$                       &                & 33.7 & 4.4   & 4.8  & 0.2\%  & 42.7 & 4.2   & 4.8   & 0.0\%  & 44.8 & 4.6    & 5.0   & 0.0\%   \\
$\widehat{\beta}_{\ipw}$                       &                & 3.7  & 15.0  & 9.9  & 89.1\% & 4.5  & 34.7  & 16.3  & 84.2\% & 4.6  & 55.1   & 22.9  & 82.7\%  \\
$\widehat{\beta}_{\aipw}$                        &                & 14.5 & 184.1 & 97.1 & 87.9\% & 31.8 & 811.4 & 371.4 & 89.6\% & 37.6 & 1223.2 & 666.2 & 95.2\%  \\
$\widehat{\beta}_{\xm}$                        &                & 4.3  & 7.0   & 8.2  & 93.9\% & 7.2  & 8.5   & 10.4  & 89.6\% & 10.3 & 8.5   & 10.2  & 84.1\%  \\
$\widehat{\beta}_{\psm.0}$                     &                & 0.6  & 8.4   & 9.9  & 96.5\% & 1.3  & 12.6  & 15.1  & 96.7\% & 2.4  & 20.4   & 20.7  & 95.4\%  \\
\multirow{5}{*}{\begin{tabular}[c]{@{}l@{}}$\widehat{\beta}_{\psm}$\\ (M=1)\end{tabular}}& software       & 4.2  & 8.2   & 9.6  & 93.7\% & 5.6  & 10.9  & 14.0  & 93.6\% & 6.9  & 17.4   & 18.7  & 92.2\%  \\
                                              & asymp          & 4.2  & 8.2   & 7.9  & 90.9\% & 5.6  & 10.9  & 10.8  & 91.3\% & 6.9  & 17.4   & 12.8  & 84.8\%  \\
                                              & naiveboot      & 4.2  & 8.2   & 9.0  & 92.7\% & 5.6  & 10.9  & 13.1  & 93.6\% & 6.9  & 17.4   & 17.4  & 92.0\%  \\
                                              & double-rsp(5)  & 4.2  & 8.2   & 9.8  & 93.7\% & 5.6  & 10.9  & 15.2  & 94.6\% & 6.9  & 17.4   & 17.7  & 91.1\%  \\
                                              & double-rsp(10) & 4.2  & 8.2   & 10.3 & 94.4\% & 5.6  & 10.9  & 18.0  & 95.2\% & 6.9  & 17.4   & 18.5  & 91.1\%  \\
\multirow{5}{*}{\begin{tabular}[c]{@{}l@{}}$\widehat{\beta}_{\psm}$\\ (M=5)\end{tabular}}& software       & 4.2  & 6.1   & 7.4  & 94.4\% & 6.5  & 8.3   & 10.2  & 92.1\% & 8.1  & 11.3   & 13.0  & 91.7\%  \\
                                              & asymp          & 4.2  & 6.1   & 5.9  & 91.0\% & 6.5  & 8.3   & 7.5   & 86.8\% & 8.1  & 11.3   & 5.0   & 82.8\%  \\
                                              & naiveboot      & 4.2  & 6.1   & 7.6  & 95.2\% & 6.5  & 8.3   & 10.4  & 92.6\% & 8.1  & 11.3   & 13.2  & 92.1\%  \\
                                              & double-rsp(5)  & 4.2  & 6.1   & 7.3  & 94.2\% & 6.5  & 8.3   & 11.2  & 93.2\% & 8.1  & 11.3   & 11.8  & 89.4\%  \\
                                              & double-rsp(10) & 4.2  & 6.1   & 7.7  & 94.9\% & 6.5  & 8.3   & 12.3  & 94.5\% & 8.1  & 11.3   & 12.9  & 90.1\%  \\
\hline
\end{tabular}}
Note: "Var" is the variance of point estimates of $\beta_{0}$ across simulated datasets;  "VE" is the average variance estimation for the point estimators over simulations, thus VE minus Var reflects the bias in estimated variance;  "CR" is the empirical coverage rate of $95\%$ confidence intervals. Five types of variance estimates for $\widehat{\beta}_{\psm}$ were compared: "software", output from the standard software; "asymp", the proposed asymptotic variance estimation; "naiveboot", the naive nonparametric bootstrap;
"double-rsp(5)", the proposed double-resampling method with five quantile strata and "double-rsp(10)", the proposed double-resampling method with ten quantile strata.
\end{table}

\begin{table}[hbt!]
\caption{\label{tab:ps=0.5}Simulation results for perfect overlap: bias ($\times10^{2}$)
and variance ($\times10^{3}$) of the point estimator of $\beta_{0}$,
coverage ($\%$) of $95\%$ confidence intervals based on $1,000$
Monte Carlo samples where the true propensity score equals 0.5 for each subject, i.e. perfect overlap}
\resizebox{\textwidth}{!}{%

\begin{tabular}{llllllllllllll} 
\hline
\multicolumn{2}{l}{}                                           & Bias & Var   & VE   & CR     & Bias & Var   & VE    & CR     & Bias & Var    & VE    & CR      \\
\multicolumn{2}{l}{}                          & \multicolumn{4}{l}{$\beta_0$=0}     & \multicolumn{4}{l}{$\beta_0$=0.5}    & \multicolumn{4}{l}{$\beta_0$=-0.5}      \\ 
\hline
\multicolumn{14}{l}{Scenario (i) Correct specification of the propensity score model}                                                                           \\
$\widehat{\beta}_{\nai}$                       &   &                            -0.1 & 4.6 & 4.8 & 95.2\% & 0.7 & 4.7 & 4.8 & 94.7\% & 0.1 & 4.5 & 4.6 & 95.7\% \\ 
$\widehat{\beta}_{\ipw}$                       &   &                            0.0 & 3.6 & 4.8 & 98.0\% & 0.8 & 3.9 & 4.8 & 96.9\% & 0.1 & 4.1 & 4.6 & 97.1\% \\ 
$\widehat{\beta}_{\aipw}$                      &   &                            0.0 & 3.3 & 3.5 & 95.7\% & 1.0 & 3.8 & 3.9 & 94.4\% & 0.1 & 3.9 & 4.2 & 96.2\% \\ 
$\widehat{\beta}_{\xm}$                        &   &                            0.0 & 4.5 & 6.3 & 98.1\% & -0.3 & 5.1 & 6.3 & 97.0\% & 0.9 & 5.4 & 6.0 & 95.5\%  \\ 
$\widehat{\beta}_{\psm.0}$                     &   &                            -0.1 & 6.0 & 6.0 & 95.6\% & 0.6 & 6.0 & 6.0 & 95.1\% & 0.1 & 5.6 & 5.7 & 94.7\% \\ 
\multirow{5}{*}{\begin{tabular}[c]{@{}l@{}}$\widehat{\beta}_{\psm}$\\ (M=1)\end{tabular}} & software       &                 0.3 & 5.5 & 6.6 & 96.5\% & 1.1 & 6.0 & 6.6 & 95.5\% & 0.4 & 6.2 & 6.3 & 95.4\% \\ 
                                                                                          & asymp          &                 0.3 & 5.5 & 5.3 & 94.3\% & 1.1 & 6.0 & 5.7 & 94.0\% & 0.4 & 6.2 & 5.9 & 94.1\% \\ 
                                                                                          & naiveboot      &                 0.3 & 5.5 & 5.8 & 95.6\% & 1.1 & 6.0 & 6.1 & 95.0\% & 0.4 & 6.2 & 6.1 & 95.2\% \\ 
                                                                                          & double-rsp(5)  &                 0.3 & 5.5 & 7.0 & 96.9\% & 1.1 & 6.0 & 6.9 & 96.2\% & 0.4 & 6.2 & 6.6 & 95.7\% \\ 
                                                                                          & double-rsp(10) &                 0.3 & 5.5 & 7.1 & 97.4\% & 1.1 & 6.0 & 7.0 & 96.6\% & 0.4 & 6.2 & 6.7 & 95.8\% \\ 
\multirow{5}{*}{\begin{tabular}[c]{@{}l@{}}$\widehat{\beta}_{\psm}$\\ (M=5)\end{tabular}} & software       &                 0.0 & 4.0 & 5.2 & 96.8\% & 0.8 & 4.4 & 5.2 & 96.1\% & 0.2 & 4.6 & 5.0 & 96.1\% \\ 
                                                                                          & asymp          &                 0.0 & 4.0 & 4.0 & 95.0\% & 0.8 & 4.4 & 4.3 & 94.3\% & 0.2 & 4.6 & 4.5 & 94.9\% \\ 
                                                                                          & naiveboot      &                 0.0 & 4.0 & 4.9 & 96.6\% & 0.8 & 4.4 & 5.1 & 96.1\% & 0.2 & 4.6 & 5.3 & 96.7\% \\ 
                                                                                          & double-rsp(5)  &                 0.0 & 4.0 & 5.4 & 97.4\% & 0.8 & 4.4 & 5.4 & 96.2\% & 0.2 & 4.6 & 5.2 & 96.1\% \\ 
                                                                                          & double-rsp(10) &                 0.0 & 4.0 & 5.5 & 97.4\% & 0.8 & 4.4 & 5.5 & 96.5\% & 0.2 & 4.6 & 5.2 & 96.3\% \\ 
\hline
\multicolumn{14}{l}{Scenario (ii) Misspecification of the propensity score model}                                                                               \\
$\widehat{\beta}_{\nai}$                       &       &                         -0.1 & 4.6 & 4.8 & 95.2\% & 0.7 & 4.7 & 4.8 & 94.7\% & 0.1 & 4.5 & 4.6 & 95.7\% \\ 
$\widehat{\beta}_{\ipw}$                       &       &                         -0.1 & 3.9 & 4.8 & 97.4\% & 0.7 & 4.2 & 4.8 & 96.8\% & 0.1 & 4.2 & 4.6 & 96.5\% \\ 
$\widehat{\beta}_{\aipw}$                        &     &                         0.0 & 3.7 & 3.9 & 95.9\% & 0.9 & 4.0 & 4.2 & 94.9\% & 0.1 & 4.0 & 4.3 & 96.1\% \\ 
$\widehat{\beta}_{\xm}$                        &       &                         0.0 & 4.8 & 6.5 & 97.9\% & -0.6 & 5.4 & 6.5 & 97.5\% & 1.3 & 5.8 & 6.2 & 95.8\% \\ 
$\widehat{\beta}_{\psm.0}$                     &       &                         -0.1 & 6.0 & 6.0 & 95.6\% & 0.6 & 6.0 & 6.0 & 95.1\% & 0.1 & 5.6 & 5.7 & 94.7\% \\
\multirow{5}{*}{\begin{tabular}[c]{@{}l@{}}$\widehat{\beta}_{\psm}$\\ (M=1)\end{tabular}}& software    &                 0.2 & 5.5 & 6.6 & 97.5\% & 0.9 & 5.9 & 6.6 & 96.7\% & 0.3 & 6.0 & 6.3 & 95.4\% \\  
                                                                                      & asymp          &                 0.2 & 5.5 & 5.5 & 95.8\% & 0.9 & 5.9 & 5.8 & 95.2\% & 0.3 & 6.0 & 6.0 & 94.7\% \\  
                                                                                      & naiveboot      &                 0.2 & 5.5 & 5.9 & 96.7\% & 0.9 & 5.9 & 6.1 & 96.1\% & 0.3 & 6.0 & 6.1 & 95.1\% \\  
                                                                                      & double-rsp(5)  &                 0.2 & 5.5 & 6.9 & 98.2\% & 0.9 & 5.9 & 6.9 & 96.8\% & 0.3 & 6.0 & 6.6 & 96.0\% \\  
                                                                                      & double-rsp(10) &                 0.2 & 5.5 & 7.0 & 98.0\% & 0.9 & 5.9 & 7.0 & 96.8\% & 0.3 & 6.0 & 6.7 & 96.2\% \\  
\multirow{5}{*}{\begin{tabular}[c]{@{}l@{}}$\widehat{\beta}_{\psm}$\\ (M=5)\end{tabular}}& software    &                 0.0 & 4.1 & 5.2 & 96.8\% & 0.7 & 4.5 & 5.2 & 96.4\% & 0.2 & 4.5 & 5.0 & 96.3\% \\  
                                                                                      & asymp          &                 0.0 & 4.1 & 4.3 & 95.3\% & 0.7 & 4.5 & 4.5 & 94.7\% & 0.2 & 4.5 & 4.6 & 95.9\% \\  
                                                                                      & naiveboot      &                 0.0 & 4.1 & 5.2 & 96.6\% & 0.7 & 4.5 & 5.4 & 96.8\% & 0.2 & 4.5 & 5.4 & 97.2\% \\  
                                                                                      & double-rsp(5)  &                 0.0 & 4.1 & 5.4 & 97.6\% & 0.7 & 4.5 & 5.4 & 96.5\% & 0.2 & 4.5 & 5.2 & 96.8\% \\  
                                                                                      & double-rsp(10) &                 0.0 & 4.1 & 5.5 & 97.5\% & 0.7 & 4.5 & 5.4 & 96.5\% & 0.2 & 4.5 & 5.2 & 96.8\% \\  
\hline
\end{tabular}}
Note: "Var" is the variance of point estimates of $\beta_{0}$ across simulated datasets;  "VE" is the average variance estimation for the point estimators over simulations, thus VE minus Var reflects the bias in estimated variance;  "CR" is the empirical coverage rate of $95\%$ confidence intervals. Five types of variance estimates for $\widehat{\beta}_{\psm}$ were compared: "software", output from the standard software; "asymp", the proposed asymptotic variance estimation; "naiveboot", the naive nonparametric bootstrap;
"double-rsp(5)", the proposed double-resampling method with five quantile strata and "double-rsp(10)", the proposed double-resampling method with ten quantile strata.
\end{table}

\end{document}